\documentclass[11pt]{article}
\usepackage{xcolor,amsmath,amssymb,graphicx}
\usepackage[english]{babel}
\usepackage[T1]{fontenc}
\usepackage{amsmath, amscd,amsthm, amsfonts, amssymb}
\usepackage{epsfig}
\usepackage{hhline, latexsym}
\usepackage{float}

\newlength{\dinwidth}
\newlength{\dinmargin}
\newcommand{\R}{\mathbb{R}}

\newcommand{\C}{\mathbb{C}}
\newcommand{\Z}{\mathbb{Z}}
\newcommand{\B}{\mathbb{B}}
\newcommand{\z}{\mathbf{z}}

\newcommand{\ud}{\,\mathrm{d}}
\newcommand{\Rs}{\mathcal{R}}

\newtheorem{theorem}{Theorem}[section]
\newtheorem{proposition}{Proposition}[section]
\newtheorem{corollary}{Corollary}[section]
\newtheorem{remark}{Remark}[section]
\newtheorem{lemma}{Lemma}[section]
\newtheorem{example}{Example}[section]

\begin{document}

\def\theequation {\thesection.\arabic{equation}}
\makeatletter\@addtoreset {equation}{section}\makeatother

\title{New degeneration of Fay's identity and\\  its application to integrable systems}

\author{
C.~Kalla\footnote{e-mail: Caroline.Kalla@u-bourgogne.fr;
address: Institut de Math\'ematiques de Bourgogne,
		Universit\'e de Bourgogne, 9 avenue Alain Savary, 21078 Dijon, France}}

\maketitle

\begin{abstract}

In this paper we prove a new degenerated version of Fay's trisecant identity. The new identity is applied to construct new algebro-geometric solutions of the multi-component nonlinear Schrödinger equation.
This approach also
provides an independent derivation of known algebro-geometric solutions to the 
Davey-Stewartson equations. 
\end{abstract}

\tableofcontents

\section{Introduction}

The well known trisecant identity discovered by Fay is a far-reaching generalization of the addition theorem
for elliptic theta functions (see \cite{Fay}). This identity states that, for any points $a,\,b,\,c,\,d$ on a compact Riemann surface of genus $g>0$, and for any $\z\in\C^{g}$, there exist constants $c_{1},\,c_{2}$ and $c_{3}$ such that
\begin{multline}
c_{1}\,\Theta\left(\z+\int^{a}_{c}\right)\Theta\left(\z+\int^{d}_{b}\right)
+c_{2}\,\Theta\left(\z+\int^{a}_{b}\right)\Theta\left(\z+\int^{d}_{c}\right)
=c_{3}\,\Theta(\z)\,\Theta\left(\z+\int^{a}_{c}+\int^{d}_{b}\right), \label{Fay intro}
\end{multline}
where $\Theta$ is the multi-dimensional theta function (\ref{2.2}); here and below we use the notation $\int^{b}_{a}$ for the Abel map (\ref{abel}) between a and b. This identity plays an important role in various domains of mathematics, as for example in the theory of Jacobian varieties \cite{Arb}, in conformal field theory \cite{Raina}, and in operator theory \cite{McC}. Moreover, as it was realized by Mumford, theta-functional 
solutions of certain integrable equations as Korteweg-de Vries (KdV), 
Kadomtsev-Petviashvili (KP), or Sine-Gordon (SG), may be derived from 
Fay's trisecant identity and its degenerations 
(see  \cite{Mum}).  

In the present paper we apply  Mumford's approach to the Davey-Stewartson equations and the multi-component nonlinear Schrödinger equation.

The first main result of this paper is a new degeneration of Fay's identity (\ref{Fay intro}).
This new identity holds for two distinct points $a,b$ on a compact Riemann surface of genus $g>0$, and any $\z\in\C^{g}$:
\begin{equation}
D'_{a}\ln\frac{\Theta(\z+\int^{b}_{a})}{\Theta(\z)}+\,D_{a}^{2}\ln\frac{\Theta(\z+\int^{b}_{a})}{\Theta(\z)}+\Big(D_{a}\ln\frac{\Theta(\z+\int^{b}_{a})}{\Theta(\z)}-K_{1}\Big)^{2}+2\,D^{2}_{a}\ln\Theta(\z)+\,K_{2}=0,  \label{deg Fay intro}
\end{equation}
where $K_{1}$ and $K_{2}$ are scalars independent of $\z$ but dependent on the points $a$ and $b$; here $D_{a}$ and $D'_{a}$ denote operators of directional derivatives along the vectors $\mathbf{V}_{a}$ and $\mathbf{W}_{a}$ (\ref{exp hol diff}). In particular, this identity implies that the following function of the variables $x$ and $t$
\begin{equation}
\psi(x,t)=A\,\frac{\Theta(\mathbf{Z}-\mathbf{d}+\int_{a}^{b})}{\Theta(\mathbf{Z}-\mathbf{d})}\,\exp\left\{ \mathrm{i}\,(-K_{1}\,x+\,K_{2}\,t)\right\}, \label{psi lin}
\end{equation}
where $\mathbf{Z}=\mathrm{i}\mathbf{V}_{a}\,x+\mathrm{i}\mathbf{W}_{a}\,t$ and $A\in\C, \mathbf{d}\in\C^{g}$ are arbitrary constants, is a solution of the linear Schrödinger equation
\begin{equation}
\mathrm{i}\,\frac{\partial \psi}{\partial t}+\frac{\partial^{2} \psi}{\partial x^{2}}+2\,u\,\psi =0,  \label{lin NLS}
\end{equation}
with the potential $u(x,t)=D^{2}_{a}\ln\Theta(\mathbf{Z})$. When this potential is related to the function $\psi$ by $u(x,t)=\rho\,|\psi|^{2}$, with  $\rho=\pm 1$, the function $\psi$ (\ref{psi lin}) becomes a solution of the nonlinear Schrödinger equation (NLS)
\begin{equation}	
\mathrm{i}\,\frac{\partial \psi}{\partial t}+\frac{\partial^{2} \psi}{\partial x^{2}}+2\rho\,|\psi|^{2}\,\psi =0.  \label{NLS}
\end{equation}
This is the starting point of our construction of algebro-geometric solutions of the Davey-Stewartson equations and the multi-component nonlinear Schrödinger equation. The nonlinear Schrödinger equation (\ref{NLS}) is a famous nonlinear dispersive partial differential equation with 
many applications, e.g.\ in hydrodynamics (deep water waves), plasma physics
and nonlinear fiber optics. Integrability of this equation was 
established by Zakharov and Shabat in \cite{ZSh}. Algebro-geometric solutions of (\ref{NLS}) were found by Its in \cite{Its}; the geometric theory of these solutions was developed by Previato \cite{P}.

There exist various ways to generalize the NLS equation.
The first is to increase the number of spatial dimensions to two. This leads to the Davey-Stewartson equations (DS), 
\begin{align}		
\mathrm{i}\,\psi_{t}+\psi_{xx}-\alpha^{2}\, \psi_{yy}+2\,(\Phi+\rho\,|\psi|^{2})\,\psi &=0,     \nonumber  \\	
\Phi_{xx}+\alpha^{2}\, \Phi_{yy}+2\rho\,|\psi|^{2}_{xx} & =0,  \label{1.2} 
\end{align}
where $\alpha=\mathrm{i},1$ and $\rho=\pm1$; $\psi(x,y,t)$ and $\Phi(x,y,t)$ are functions of the real variables $x,\,y$ and $t$, the latter being real valued and the former being complex valued. In what follows, DS1$^{\rho}$ denotes the Davey-Stewartson 
equation when $\alpha=\mathrm{i}$, and DS2$^{\rho}$ the Davey-Stewartson 
equation when $\alpha=1$. 
The Davey-Stewartson equation (\ref{1.2}) was introduced in \cite{DS} to describe the evolution of a three-dimensional wave package on water of finite depth. Complete 
integrability of the equation was shown in \cite{AF}. If solutions $\psi$ and $\Phi$ of (\ref{1.2}) do not depend on the variable $y$ the first equation in (\ref{1.2}) reduces to the NLS equation (\ref{NLS})  under appropriate boundary conditions for the function $\Phi+\rho\,|\psi|^{2}$ in the limit when $x$ tends to infinity.

Algebro-geometric solutions of the Davey-Stewartson equations (\ref{1.2}) were previously obtained in \cite{Mal} using the formalism of Baker-Akhiezer functions. In both \cite{Mal} and the present paper, solutions of (\ref{1.2}) are constructed from solutions of the complexified 
system which, after the change of coordinates $\xi=\frac{1}{2}(x-\mathrm{i}\alpha y)$ and $\eta=\frac{1}{2}(x+\mathrm{i}\alpha 
y),$ with $\alpha=\mathrm{i},1$, reads
\begin{align}
	\mathrm{i}\,\psi_{t}+\frac{1}{2}(\psi_{\xi\xi}+\psi_{\eta\eta})+2\,\varphi\,\psi=0, \nonumber\\ \label{DS intro}
	-\mathrm{i}\,\psi^{*}_{t}+\frac{1}{2}(\psi^{*}_{\xi\xi}+\psi^{*}_{\eta\eta})+2\,\varphi\,\psi^{*}=0,\\ \nonumber
\varphi_{\xi\eta}+\frac{1}{2}((\psi\psi^{*})_{\xi\xi}+(\psi\psi^{*})_{\eta\eta})=0,
\end{align}
where $\varphi:=\Phi+\psi\psi^{*}$. This system reduces to (\ref{1.2}) under the reality condition: 
\begin{equation}
\psi^{*}=\rho\,\overline{\psi}. \label{real cond DS intro}
\end{equation}
The second main result of our paper is an independent derivation of the solutions  \cite{Mal} using the degenerated Fay identity (\ref{deg Fay intro}). Algebro-geometric data associated to these solutions are $\{\Rs_{g},a,b,k_{a},k_{b}\}$, where $\Rs_{g}$ is a compact Riemann surface of genus $g>0$, $a$ and $b$ are two distinct points on $\Rs_{g}$, and $k_{a},k_{b}$ are arbitrary local parameters near $a$ and $b$. These solutions read
\begin{align}
\psi(\xi,\eta,t)&=A\,\frac{\Theta(\mathbf{Z}-\mathbf{d}+\int^{b}_{a})}{\Theta(\mathbf{Z}-\mathbf{d})}\,\exp\left\{ \mathrm{i}\left(-G_{1}\,\xi-G_{2}\,\eta+G_{3}\,\tfrac{t}{2}\right)\right\},\nonumber\\
\psi^{*}(\xi,\eta,t)&=-\frac{\kappa_{1}\kappa_{2}\,q_{2}(a,b)}{A}\,\frac{\Theta(\mathbf{Z}-\mathbf{d}-\int^{b}_{a})}{\Theta(\mathbf{Z}-\mathbf{d})}\,\exp\left\{ \mathrm{i}\left(G_{1}\,\xi+G_{2}\,\eta-G_{3}\,\tfrac{t}{2}\right)\right\},\nonumber\\
\varphi(\xi,\eta,t)&=\frac{1}{2}(\ln\,\Theta(\mathbf{Z-\mathbf{d}}))_{\xi\xi}+\frac{1}{2}(\ln\,\Theta(\mathbf{Z-\mathbf{d}}))_{\eta\eta}+\frac{1}{4}h, \nonumber
\end{align}
where the scalars $G_{i},\,q_{2}(a,b)$ depend on the points $a,\,b\in\Rs_{g}$, and $\kappa_{1},\kappa_{2},A,h\in\C,$ $\mathbf{d}\in\C^{g}$ are arbitrary constants; the $g$-dimensional vector $\mathbf{Z}$ is a linear function of the variables $\xi,\,\eta$ and $t$.
The reality condition (\ref{real cond DS intro}) imposes constraints on the associated algebro-geometric data. In particular, the Riemann surface $\Rs_{g}$ has to be real. The approach used in \cite{Mal} to study reality conditions (\ref{real cond DS intro}) is based on properties of Baker-Akhiezer functions. Our present approach based on identity (\ref{deg Fay intro}) allows to construct solutions of DS1$^{\rho}$ and DS2$^{\rho}$ corresponding to Riemann surfaces of more general topological type than in \cite{Mal}.

Another way to generalize the NLS equation is to increase the number 
of dependent variables in (\ref{NLS}). This leads to the multi-component nonlinear Schrödinger equation
\begin{equation}		
\mathrm{i}\,\frac{\partial \psi_{j}}{\partial t}+\frac{\partial^{2} \psi_{j}}{\partial x^{2}}+2\,\left(\sum_{k=1}^{n}s_{k}|\psi_{k}|^{2}\right)\,\psi_{j} =0,  \quad \quad j=1,\ldots,n,   \label{n-NLS^{s}}
\end{equation}
denoted by n-NLS$^{s}$,  where 
$s=(s_{1},\ldots,s_{n})$,  $s_{k}=\pm 1$. Here $\psi_{j}(x,t)$ are complex valued functions of the real variables $x$ and $t$. The case $n=1$ corresponds to the NLS 
equation. 
The integrability of the two-component nonlinear 
Schrödinger equation (\ref{n-NLS^{s}}) in the case $s=(1,1)$ was first established by Manakov \cite{Man}; integrability for the multi-component case with any $n\geq 2$ and $s_{k}=\pm 1$ was established in \cite{RSL}.
Algebro-geometric solutions of the two-component NLS equation with signature $(1,1)$ were investigated in \cite{EEI} using the Lax formalism and Baker-Akhiezer functions; these solutions are expressed in terms of theta functions of special trigonal spectral curves.

The third main result of this paper is the construction of smooth algebro-geometric solutions of the multi-component nonlinear Schrödinger equation (\ref{n-NLS^{s}}) for arbitrary $n\geq 2$, obtained by using (\ref{deg Fay intro}). We first find solutions to the complexified system
\begin{align}		
\mathrm{i}\,\frac{\partial \psi_{j}}{\partial t}+\frac{\partial^{2} \psi_{j}}{\partial x^{2}}+2\,\left(\sum_{k=1}^{n}\,\psi_{k}\,\psi_{k}^{*}\right)\,\psi_{j} &=0  ,  \nonumber  \\ 
-\mathrm{i}\,\frac{\partial \psi_{j}^{*}}{\partial t}+\frac{\partial^{2} \psi_{j}^{*}}{\partial x^{2}}+2\,\left(\sum_{k=1}^{n}\,\psi_{k}\,\psi_{k}^{*}\right)\,\psi_{j}^{*} &=0 , \quad \quad j=1,\ldots,n  ,\label{ass n-NLS intro} 
\end{align}
where $\psi_{j}(x,t)$ and $\psi_{j}^{*}(x,t)$ are complex valued functions of the real variables $x$ and $t$. This system reduces to the n-NLS$^{s}$ equation (\ref{n-NLS^{s}}) under the reality conditions 
\begin{equation}
\psi_{j}^{*}=s_{j}\,\overline{\psi_{j}}, \quad\quad j=1,\ldots, n. \label{real cond intro}
\end{equation}
Algebro-geometric data associated to the solutions of (\ref{ass n-NLS intro}) are given by $\{\Rs_{g},f,z_{a}\}$, where $\Rs_{g}$ is a compact Riemann surface of genus $g>0$, $f$ is a meromorphic function of degree $n+1$ on $\Rs_{g}$ and $z_{a}\in\C\mathbb{P}^{1}$ is a non critical value of the meromorphic function $f$ such that $f^{-1}(z_{a})=\{a_{1},\ldots,a_{n+1}\}$. Then the solutions $\left\{\psi_{j}\right\}_{j=1}^{n}$ and  $\left\{\psi_{j}^{*}\right\}_{j=1}^{n}$ of system (\ref{ass n-NLS intro}) read
\begin{align}    
\psi_{j}(x,t)&=A_{j}\,\frac{\Theta(\mathbf{Z}-\mathbf{d}+\int_{a_{n+1}}^{a_{j}})}{\Theta(\mathbf{Z}-\mathbf{d})}\,\exp\left\{ \mathrm{i}\,(-E_{j}\,x+\,F_{j}\,t)\right\},\nonumber\\ 
\psi^{*}_{j}(x,t)&=\frac{q_{2}(a_{n+1},a_{j})}{A_{j}}\,\frac{\Theta(\mathbf{Z}-\mathbf{d}-\int_{a_{n+1}}^{a_{j}})}{\Theta(\mathbf{Z}-\mathbf{d})}\,\exp\left\{ \mathrm{i}\,(E_{j}\,x-\,F_{j}\,t)\right\}, \nonumber
\end{align}
where the scalars $E_{j},\,F_{j},\,q_{2}(a_{n+1},a_{j})$ depend on the points $a_{n+1},\,a_{j}\in\Rs_{g}$, and $A_{j}\in\C,$ $\mathbf{d}\in\C^{g}$ are arbitrary constants; here the $g$-dimensional vector $\mathbf{Z}$ is a linear function of the variables $x$ and $t$. Imposing the reality conditions (\ref{real cond intro}), we describe explicitly solutions for the focusing case $s=(1,\ldots,1)$ and the defocusing case $s=(-1,\ldots,-1)$ associated to a real branched covering of the Riemann sphere. In particular, our solutions of the focusing case are associated to a covering without real branch points. Our general construction, being applied to
the two-component case, gives solutions with more parameters than in \cite{EEI} for fixed
genus of the spectral curve. Moreover, we provide smoothness conditions for our solutions.

The paper is  organized as follows: in section 2 we recall some facts 
about the theory of Riemann surfaces, and derive a new degeneration of Fay's identity. With this degeneration, 
 we give  in Section 3 an independent derivation of smooth theta-functional solutions of the Davey Stewartson equations; this approach also provides an explicit description of the constants appearing in the solutions in terms of theta functions. In Section 4, we construct new smooth theta-functional solutions of the multi-component NLS equation, and describe explicitely solutions of the focusing and defocusing cases. We also discuss the reduction from n-NLS to (n-1)-NLS, stationary solutions of n-NLS, and the link between solutions of n-NLS and solutions of the KP1 equation.
Appendix A contains various facts from the theory of real Riemann surfaces. Appendix B contains an auxiliary computation required in the construction of algebro-geometric solutions of DS and n-NLS equations.

\section{New degeneration of Fay's identity}

In this section we recall some facts from the classical theory of 
Riemann surfaces \cite{Fay} and derive a new corollary of Fay's 
trisecant identity. 

\subsection{Theta functions} Let $\Rs _{g}$ be a 
compact Riemann surface of genus $g>0$. Denote by $(\mathcal{A}_{1},\ldots,\mathcal{A}_{g},\mathcal{B}_{1},\ldots,\mathcal{B}_{g})$ a canonical homology basis, and by $(\omega_{1},\ldots,\omega_{g})$ the dual basis of holomorphic differentials normalized via
\begin{equation}
\int_{\mathcal{A}_{k}}\omega_{j}=2\mathrm{i}\pi\delta_{jk}, \quad j,k=1,\ldots,g. \label{norm hol diff}
\end{equation}
The matrix $\B=\left(\int_{\mathcal{B}_{k}}\omega_{j}\right)$ of $\mathcal{B}$-periods of the normalized holomorphic differentials $\left\{\omega_{j}\right\}_{j=1}^{g}$
is symmetric and has a negative definite real part. The theta function with (half integer) characteristics $\delta=[\delta' \delta'']$ is defined by
\begin{equation}
\Theta[\delta](\z|\B)=\sum_{\mathbf{m}\in\Z^{g}}\exp\left\{\tfrac{1}{2}\langle \B(\mathbf{m}+\delta'),\mathbf{m}+\delta'\rangle+\langle \mathbf{m}+\delta',\z+2\mathrm{i}\pi\delta''\rangle\right\};\label{2.2}
\end{equation}
here $\z\in\C^{g}$ is the argument and $\delta',\delta''\in \left\{0,\frac{1}{2}\right\}^{g}$ 
are the vectors of characteristics; $\langle,\rangle$ denotes the 
scalar product $\left\langle \mathbf{u},\mathbf{v} \right\rangle=\sum_{i}u_{i}\,v_{i}$ for any $\mathbf{u},\,\mathbf{v}\in\C^{g}$. The theta function $\Theta[\delta](\z)$ is even if 
the characteristic $\delta$ is even i.e, $4\left\langle 
\delta',\delta'' \right\rangle$ is even, and odd if the characteristic 
$\delta$ is odd i.e., $4\left\langle \delta',\delta'' 
\right\rangle$ is odd. An even characteristic is called nonsingular if 
$\Theta[\delta](\textbf{0})\neq 0$, and an odd characteristic is called nonsingular if the gradient $\nabla\Theta[\delta](\textbf{0})$ is non-zero. The theta function with characteristics is related to 
the theta function with zero characteristics (denoted by $\Theta$) as 
follows
\begin{equation}
\Theta[\delta](\z)=\Theta(\z+2\mathrm{i}\pi\delta''+\B\delta')\,\exp\left\{\tfrac{1}{2}\langle \B\delta',\delta'\rangle+\langle\z+2\mathrm{i}\pi\delta'',\delta'\rangle\right\}.\label{2.3}
\end{equation}
Let $\Lambda$ be the lattice $\Lambda=\{2\mathrm{i}\pi \mathbf{N}+\B 
\mathbf{M}, \,\,\mathbf{N},\mathbf{M}\in\Z^{g}\}$ generated by the 
$\mathcal{A}$ and $\mathcal{B}$-periods of the normalized holomorphic differentials $\left\{\omega_{j}\right\}_{j=1}^{g}$. The complex torus $J=J(\Rs_{g})=\C^{g} / \Lambda$ is called the Jacobian of the Riemann surface $\Rs_{g}$. The theta function with characteristics (\ref{2.2}) has the following quasi-periodicity property
\begin{equation}
\Theta[\delta](\mathbf{z}+2\mathrm{i}\pi \mathbf{N}+\B \mathbf{M}) \nonumber
\end{equation}
\begin{equation}
=\Theta[\delta](\z)  \exp\left\{-\tfrac{1}{2}\langle \B\mathbf{M},\mathbf{M}\rangle-\langle \z,\mathbf{M}\rangle+2\mathrm{i}\pi(\langle\delta',\mathbf{N}\rangle-\langle\delta'',\mathbf{M}\rangle)\right\}.\label{2.4}
\end{equation}
Denote by $\mu$ the Abel map $\mu:\Rs_{g}\longmapsto J$ defined by 
\begin{equation}
\mu(p)=\int_{p_{0}}^{p}\omega,  \label{abel}
\end{equation} 
for any $p\in\Rs_{g}$, where $p_{0}\in\Rs_{g}$ is the base point of the application, and $\omega=(\omega_{1},\ldots,\omega_{g})$ is the vector of the normalized holomorphic differentials. In the whole paper we use the notation $\int_{a}^{b}=\mu(b)-\mu(a)$.

\vspace{0.5cm}
\subsection{Fay's identity and  previously known degenerations} Let us introduce the prime-form which is given by
\begin{equation}
E(a,b)=\frac{\Theta[\delta](\int_{b}^{a})}{h_{\delta}(a)h_{\delta}(b)}, \label{prime}
\end{equation}
$a,\,b\in\Rs_{g}$; $h_{\delta}(a)$ is a spinor defined by $h^{2}_{\delta}(a)=\sum_{j=1}^{g}\frac{\partial\Theta[\delta]}{\partial z_{j}}(0)\omega_{j}(a)$, where $\delta=[\delta' \delta'']$ is a non-singular odd 
 characteristic (the prime form is independent of the choice of the characteristic $\delta$). 
Fay's trisecant identity has the form 
\begin{multline}
E(a,b)E(c,d)\Theta\left(\z+\int^{a}_{c}\right)\Theta\left(\z+\int^{d}_{b}\right)\\
+E(a,c)E(d,b)\,\Theta\left(\z+\int^{a}_{b}\right)\Theta\left(\z+\int^{d}_{c}\right)\\
=E(a,d)E(c,b)\,\Theta(\z)\,\Theta\left(\z+\int^{a}_{c}+\int^{d}_{b}\right),\label{2.5}
\end{multline}
where $a,b,c,d\in\Rs_{g}$ and all integration contours do not intersect cycles of the canonical 
homology basis. Let us now discuss degenerations of identity (\ref{2.5}).

Let $k_{a}(p)$ denote a local parameter near $a\in\Rs_{g}$, where $p$ lies in a neighbourhood of $a$. Consider 
the following expansion of the normalized holomorphic differentials $\omega_{j}$  near $a$,
\begin{equation} 
\omega_{j}(p)= \left(V_{a,j}+W_{a,j}\,k_{a}(p)+U_{a,j}\,\frac{k_{a}(p)^{2}}{2!}+\ldots\right)\,\ud k_{a}(p), \label{exp hol diff}
\end{equation}
where $V_{a,j},\,W_{a,j},\,U_{a,j}\in\C$.
Let us denote by $D_{a}$ the operator of directional derivative along the vector $\mathbf{V}_{a}=(V_{a,1},\ldots,V_{a,g})$:
\begin{equation} 
D_{a}F(\z)=\sum_{j=1}^{g}\partial_{z_{j}}F(\z) V_{a,j}=\left\langle \nabla F(\z),\mathbf{V}_{a}\right\rangle, \label{2.7}
\end{equation}
where $F:\C^{g}\longrightarrow \C$ is an arbitrary function, 
and denote by $D'_{a}$ the operator of directional derivative along the vector $\mathbf{W}_{a}=(W_{a,1},\ldots,W_{a,g})$:
\[D'_{a}F(\z)=\sum_{j=1}^{g}\partial_{z_{j}}F(\z) W_{a,j}=\left\langle \nabla F(\z),\mathbf{W}_{a}\right\rangle.\]
Then for any $\z\in\C^{g}$ and any distinct points $a,b\in\Rs_{g}$, the following well-known degenerated version of Fay's identity holds (see \cite{Mum})
\begin{equation} 
D_{a}D_{b}\,\ln\,\Theta(\z)\,=\,q_{1}(a,b)\,+\,q_{2}(a,b)\,\frac{\Theta(\z+\int^{b}_{a})\,\Theta(\z+\int^{a}_{b})}{\Theta(\z)^{2}},
\label{cor Fay}
\end{equation}
where the scalars $q_{1}(a,b)$ and $q_{2}(a,b)$ are given by
\begin{align}
q_{1}(a,b)&=D_{a}D_{b}\ln\Theta[\delta](\int^{b}_{a}), \label{q1}\\ 
q_{2}(a,b)&=\frac{D_{a}\,\Theta[\delta](0)\,D_{b}\,\Theta[\delta](0)}{\Theta[\delta](\int^{b}_{a})^{2}},\label{q2}
\end{align}
where $\delta$ is a non-singular odd  characteristic. 
Notice that $q_{1}(a,b)$ and $q_{2}(a,b)$ depend on the choice of local parameters $k_{a}$ and $k_{b}$ near $a$ and $b$ respectively.

\subsection{New degeneration of Fay's identity}

Algebro-geometric solutions of the Davey-Stewartson equations and the multi-component NLS equation constructed in this paper are obtained by using the following new degenerated version of Fay's identity.

\begin{theorem}
Let $a,b$ be distinct points on a compact Riemann surface $\Rs_{g}$ of genus $g$. Fix local parameters $k_{a}$ and $k_{b}$ in a neighbourhood of $a$ and $b$ respectively. Denote by $\delta$ a non-singular odd 
 characteristic. Then for any $\z\in \C^{g}$,
\begin{equation}
D'_{a}\ln\frac{\Theta(\z+\int^{b}_{a})}{\Theta(\z)}+\,D_{a}^{2}\ln\frac{\Theta(\z+\int^{b}_{a})}{\Theta(\z)}+\Big(D_{a}\ln\frac{\Theta(\z+\int^{b}_{a})}{\Theta(\z)}-K_{1}(a,b)\Big)^{2}+2\,D^{2}_{a}\ln\Theta(\z)+\,K_{2}(a,b)=0,\label{my corol}
\end{equation}
where the scalars $K_{1}(a,b)$ and $K_{2}(a,b)$ are given by
\begin{equation}
K_{1}(a,b)=\frac{1}{2}\frac{\,D_{a}'\,\Theta[\delta](0)}{D_{a}\,\Theta[\delta](0)}+\,D_{a}\ln\Theta[\delta](\int^{b}_{a})\,, \label{K1}
\end{equation}
and
\begin{equation}
K_{2}(a,b)=-D'_{a}\ln\Theta(\int^{b}_{a})-\,D_{a}^{2}\ln\left(\Theta(\int^{b}_{a})\Theta(0)\right)-\Big(D_{a}\ln\Theta(\int^{b}_{a})-K_{1}(a,b)\Big)^{2}. \label{K2}
\end{equation}
\end{theorem}

\begin{proof}
We start from the following lemma
\\\\
\begin{lemma}
Let $b,c\in\Rs_{g}$ be distinct points. Fix local parameters $k_{b}$ and $k_{c}$ in a neighbourhood of $b$ and $c$ respectively. Then for any $\z\in \C^{g}$, 
\begin{multline}
D_{c}\left[-D'_{b}\ln\frac{\Theta(\z+\int^{b}_{c})}{\Theta(\z)}+\,D_{b}^{2}\ln\frac{\Theta(\z+\int^{b}_{c})}{\Theta(\z)}
+\left(D_{b}\ln\frac{\Theta(\z+\int^{b}_{c})}{\Theta(\z)}+K_{1}(b,c)\right)^{2}+2\,D^{2}_{b}\ln\Theta(\z)\right]=0,  \label{2.14} 
\end{multline} 
where the scalar $K_{1}(b,c)$ is defined in (\ref{K1}). 
\end{lemma}
\noindent
\textit{Proof of Lemma 2.1.} Let us introduce the notations
    $\Theta_{ab}=\Theta(\z+\int_{a}^{b}\omega)$ and 
    $\Theta=\Theta(\mathbf{z})$. 
Differentiating (\ref{2.5}) twice with respect to the local parameter $k_{d}(p)$, where $p$ lies in a neighbourhood of $d$, and taking the limit $d\rightarrow b$, we obtain
\begin{equation}
D'_{b}\ln\Theta+D_{b}^{2}\ln\Theta+(D_{b}\ln\Theta)^{2}+\frac{p_{3}}{p_{2}}\,D_{b}\ln\Theta_{ca}-\frac{p_{3}}{p_{2}}D_{b}\ln\Theta  \label{2.15}
\end{equation}
\[=\frac{p_{1}p_{3}}{p_{2}}-2\,D_{b}\ln\Theta_{ca}\,D_{b}\ln\Theta_{cb}+2\,D_{b}\ln\Theta\,D_{b}\ln\Theta_{cb}+2\,p_{1}\,D_{b}\ln\Theta_{cb}\]
\[-p_{4}-2\,p_{1}\,D_{b}\ln\Theta_{ca}+D'_{b}\ln\Theta_{ca}+D_{b}^{2}\ln\Theta_{ca}+(D_{b}\ln\Theta_{ca})^{2},\]
where we took into account the relation
\[\partial_{k_{d}}^{2}\,\Theta(\z+\int^{d}_{b})\big|_{d=b} \,=\,D'_{b}\Theta(\z)+D_{b}^{2}\Theta(\z).\]
The quantities $p_{j}=p_{j}(a,b,c),$ for $j=1,2,3,4,$ are given by
\begin{equation}
p_{1}(a,b,c)=-\frac{E(c,b)}{E(a,b)}\partial_{k_{x}}\frac{E(a,x)}{E(c,x)}\Big |_{x=b}, \quad p_{2}(a,b,c)=\frac{E(a,c)}{E(a,b)}\partial_{k_{x}}\frac{E(x,b)}{E(c,x)}\Big |_{x=b}, \label{p2}
\end{equation}
\begin{equation}
p_{3}(a,b,c)=\frac{E(a,c)}{E(a,b)}\partial^{2}_{k_{x}}\frac{E(x,b)}{E(c,x)}\Big |_{x=b} \,, \quad p_{4}(a,b,c)=-\frac{E(c,b)}{E(a,b)}\partial^{2}_{k_{x}}\frac{E(a,x)}{E(c,x)}\Big |_{x=b}. \label{p3}
\end{equation}
\\
Differentiating (\ref{2.15}) with respect to the local parameter $k_{a}(p)$, where $p$ lies in a neighbourhood of $a$, and taking the limit $a\rightarrow c$, we get
\begin{multline}
D_{c}D'_{b}\ln\Theta+D_{c}D_{b}^{2}\ln\Theta-2\,D_{c}D_{b}\ln\Theta\,D_{b}\ln\frac{\Theta_{cb}}{\Theta}
+2\,q_{1}\,D_{b}\ln\frac{\Theta_{cb}}{\Theta}-\,\frac{p_{3}}{p_{2}}\,D_{c}D_{b}\ln\Theta+K=0, \label{2.16}
\end{multline}
where the scalar $K$ depends on the points $b,c$, but not on the vector $\z\in\C^{g}$. Here the scalars $q_{1},p_{2}$ and $p_{3}$ are defined in (\ref{q1}), (\ref{p2}), and (\ref{p3}) respectively. The change of variable $\z\leftrightarrow -\z+\int^{c}_{b}$ in (\ref{2.16}) leads to
\begin{multline}
D_{c}D'_{b}\ln\Theta_{cb}-D_{c}D_{b}^{2}\ln\Theta_{cb}-2\,D_{c}D_{b}\ln\Theta_{cb}\,D_{b}\ln\frac{\Theta_{cb}}{\Theta}
+2\,q_{1}\,D_{b}\ln\frac{\Theta_{cb}}{\Theta}-\,\frac{p_{3}}{p_{2}}\,D_{c}D_{b}\ln\Theta_{cb}+K=0. \label{2.17}
\end{multline}
Now (\ref{2.14}) is obtained by substracting (\ref{2.16}) and (\ref{2.17}). $\hspace{7cm} \square$
\\\\

To proof Theorem 2.1, make the change of variable $\z\mapsto 
-\z+\int^{c}_{b}$ in (\ref{2.15}) and add $2\,D_{b}^{2}\ln\Theta$ to 
each side of the equality to get
\[-D'_{b}\ln\frac{\Theta_{cb}}{\Theta}+D_{b}^{2}\ln\frac{\Theta_{cb}}{\Theta}+\left(D_{b}\ln\frac{\Theta_{cb}}{\Theta}+\frac{1}{2}\frac{p_{3}}{p_{2}}\right)^{2}-\frac{1}{4}\left(\frac{p_{3}}{p_{2}}\right)^{2}-\frac{p_{1}p_{3}}{p_{2}}+2\,D_{b}^{2}\ln\Theta\]
\[=-D'_{b}\ln\frac{\Theta_{ab}}{\Theta}+D_{b}^{2}\ln\frac{\Theta_{ab}}{\Theta}+\left(D_{b}\ln\frac{\Theta_{ab}}{\Theta}+\frac{1}{2}\left(\frac{p_{3}}{p_{2}}+2p_{1}\right)\right)^{2}-\frac{1}{4}\Big(\frac{p_{3}}{p_{2}}+2\,p_{1}\Big)^{2}-p_{4}+2\,D_{b}^{2}\ln\Theta.\]
By Lemma 2.1, the directional derivative of the left hand side of the previous equality along the vector $\mathbf{V}_{c}$ equals zero. Hence for any distinct points $a,b,c\in\Rs_{g}$, we get
\begin{equation}
D_{c}\Big[-D'_{b}\ln\frac{\Theta_{ab}}{\Theta}+D_{b}^{2}\ln\frac{\Theta_{ab}}{\Theta}+\left(D_{b}\ln\frac{\Theta_{ab}}{\Theta}+\frac{1}{2}\left(\frac{p_{3}}{p_{2}}+2p_{1}\right)\right)^{2}+2\,D_{b}^{2}\ln\Theta\Big]=0. \label{P21 1}
\end{equation}
Moreover, from (\ref{p2}), (\ref{p3}) and (\ref{prime}), it can be seen that 
the expression
$\frac{1}{2}(\frac{p_{3}}{p_{2}}+2p_{1})$ does not depend on the point 
$c$ and equals $K_{1}(b,a)$ given by (\ref{K1}). Now let us introduce the following function of the variable $\z\in \C^{g}$  \[f_{(b,a)}(\z)=-D'_{b}\ln\frac{\Theta_{ab}}{\Theta}+D_{b}^{2}\ln\frac{\Theta_{ab}}{\Theta}+\Big(D_{b}\ln\frac{\Theta_{ab}}{\Theta}+K_{1}(b,a)\big)\Big)^{2}+2\,D_{b}^{2}\ln\Theta.\] 
Then (\ref{P21 1}) can be rewritten as $D_{c}\,f_{(b,a)}(\z)=0$ for any
$\z\in \C^{g}$ and for all $c\in\Rs_{g}$, $c\neq b$ (because also $D_{a}\,f_{(b,a)}(\z)=0$ by Lemma 2.1). 
Due to the fact that on each Riemann surface $\Rs_{g}$, there exists a positive divisor $d_{1}+...+d_{g}$ of degree $g$ such that vectors $\frac{\omega(d_{1})}{\ud k_{d_{1}}},...,\frac{\omega(d_{g})}{\ud k_{d_{g}}}$ are linearly independent (see \cite{KKV}, Lemma 5), the function $f_{(b,a)}(\z)$ is constant with 
respect to $\z$; we denote this constant by $-K_{2}(b,a)$:
\begin{equation}
f_{(b,a)}(\z)=-K_{2}(b,a) \label{P21 2}
\end{equation}
for any $\z\in \C^{g}$.
Interchanging $a$ and $b$, and changing the variable $\z\leftrightarrow -\z$ in (\ref{P21 2}) we get (\ref{my corol}). The expression (\ref{K2}) for the scalar $K_{2}(a,b)$ follows from 
(\ref{P21 2}) putting $\z=0$.  
\end{proof}

\section{Algebro-geometric solutions of the Davey-Stewartson equations}

Here we derive algebro-geometric solutions of the Davey-Stewartson equations 
(\ref{1.2}) using the degeneration (\ref{my corol}) of Fay's identity. 
Let us introduce the function $\phi:=\Phi+\rho|\psi|^{2}$, where $\rho=\pm1$,  and the differential operators 
\[D_{1}=\partial_{xx}-\alpha^{2}\partial_{yy}, \quad 
D_{2}=\partial_{xx}+\alpha^{2}\partial_{yy}.\]
Introduce also the characteristic coordinates 
\[\xi=\frac{1}{2}(x-\mathrm{i}\alpha y),\quad \eta=\frac{1}{2}(x+\mathrm{i}\alpha 
y),\quad \alpha=\mathrm{i},1.\]
In these coordinates the Davey Stewartson equations (\ref{1.2}) become
\begin{align}	
\mathrm{i}\,\psi_{t}+D_{1}\psi+2\,\phi\,\psi &=0,     \nonumber  \\
D_{2}\phi+\rho \,D_{1}|\psi|^{2} & =0,  \label{4.1} 
\end{align}
where the differential operators $D_{1}$ and $D_{2}$ are given by
\[D_{1}=\frac{1}{2}(\partial_{\xi}^{2}+\partial_{\eta}^{2}),\quad D_{2}=\partial_{\xi}\partial_{\eta}.\]
In what follows, DS1$^{\rho}$ denotes the Davey-Stewartson 
equation when $\alpha=\mathrm{i}$ (in this case $\xi$ and $\eta$ are both 
real), and DS2$^{\rho}$ the Davey-Stewartson 
equation when $\alpha=1$ (in this case $\xi$ and $\eta$ are pairwise conjugate).

\subsection{Solutions of the complexified Davey-Stewartson equations} 

To construct algebro-geometric solutions of (\ref{4.1}), let us first introduce the complexified Davey-Stewartson equations
\begin{align}
	\mathrm{i}\,\psi_{t}+\frac{1}{2}(\psi_{\xi\xi}+\psi_{\eta\eta})+2\,\varphi\,\psi=0, \nonumber\\ \label{DS}
	-\mathrm{i}\,\psi^{*}_{t}+\frac{1}{2}(\psi^{*}_{\xi\xi}+\psi^{*}_{\eta\eta})+2\,\varphi\,\psi^{*}=0,\\ \nonumber
\varphi_{\xi\eta}+\frac{1}{2}((\psi\psi^{*})_{\xi\xi}+(\psi\psi^{*})_{\eta\eta})=0,
\end{align}
where $\varphi:=\Phi+\psi\psi^{*}$. This system reduces to (\ref{4.1}) under the \textit{reality condition}: 
\begin{equation}
\psi^{*}=\rho\,\overline{\psi}, \label{real cond DS}
\end{equation}
which leads to $\varphi=\phi$. Theta functional solutions of system (\ref{DS}) are given by
\begin{theorem}
Let $\mathcal{R}_{g}$ be a compact Riemann surface of genus $g>0$, and let $a,b\in\Rs_{g}$ be distinct points. Take arbitrary constants $\mathbf{d}\in\C^{g}$ and $A,\,\kappa_{1},\,\kappa_{2}\in\C\setminus\left\{0\right\},\,h\in\C$. Denote by $\ell$ a contour connecting $a$ and $b$ which does not intersect cycles of the canonical homology basis. Then for any $\xi,\,\eta,\,t\in\C$, the following functions $\psi$, $\psi^{*}$ and $\varphi$ are solutions of system (\ref{DS})
\begin{align}
\psi(\xi,\eta,t)&=A\,\frac{\Theta(\mathbf{Z}-\mathbf{d}+\mathbf{r})}{\Theta(\mathbf{Z}-\mathbf{d})}\,\exp\left\{ \mathrm{i}\left(-G_{1}\,\xi-G_{2}\,\eta+G_{3}\,\tfrac{t}{2}\right)\right\},\nonumber\\
\psi^{*}(\xi,\eta,t)&=-\frac{\kappa_{1}\kappa_{2}\,q_{2}(a,b)}{A}\,\frac{\Theta(\mathbf{Z}-\mathbf{d}-\mathbf{r})}{\Theta(\mathbf{Z}-\mathbf{d})}\,\exp\left\{ \mathrm{i}\left(G_{1}\,\xi+G_{2}\,\eta-G_{3}\,\tfrac{t}{2}\right)\right\},\label{sol DS}\\
\varphi(\xi,\eta,t)&=\frac{1}{2}(\ln\,\Theta(\mathbf{Z-\mathbf{d}}))_{\xi\xi}+\frac{1}{2}(\ln\,\Theta(\mathbf{Z-\mathbf{d}}))_{\eta\eta}+\frac{1}{4}h. \nonumber
\end{align}
Here $\mathbf{r}=\int_{\ell}\omega$, where $\omega$ is the vector of normalized holomorphic differentials, and
\begin{equation}
\mathbf{Z}=\mathrm{i}\left(\kappa_{1}\,\mathbf{V}_{a}\,\xi-\kappa_{2}\,\mathbf{V}_{b}\,\eta+(\kappa^{2}_{1}\,\mathbf{W}_{a}-\kappa^{2}_{2}\,\mathbf{W}_{b})\,\tfrac{t}{2}\right), \label{Z DS}
\end{equation}
where the vectors $\mathbf{V}_{a},\,\mathbf{V}_{b}$ and $\mathbf{W}_{a},\,\mathbf{W}_{b}$ were introduced in (\ref{exp hol diff}). The scalars $G_{1},G_{2},G_{3}$ are given by
\begin{equation}
G_{1}=\kappa_{1}\,K_{1}(a,b),\qquad G_{2}=\kappa_{2}\,K_{1}(b,a), \label{N12 DS}
\end{equation}
\begin{equation}
G_{3}=\kappa^{2}_{1}\,K_{2}(a,b)+\kappa^{2}_{2}\,K_{2}(b,a)+h, \label{N3 DS}
\end{equation}
and scalars $q_2(a,b), \,K_{1}(a,b),\,K_{2}(a,b)$ are defined in (\ref{q2}), (\ref{K1}), (\ref{K2}) respectively. 
\end{theorem}

\begin{proof} Substitute functions (\ref{sol DS}) in the first equation of system (\ref{DS}) to get
\[\kappa^{2}_{1}\, D'_{a}\ln\frac{\Theta(\mathbf{Z}-\mathbf{d}+\mathbf{r})}{\Theta(\mathbf{Z}-\mathbf{d})}+\kappa^{2}_{1}\,D_{a}^{2}\ln\frac{\Theta(\mathbf{Z}-\mathbf{d}+\mathbf{r})}{\Theta(\mathbf{Z}-\mathbf{d})}+2\,\kappa^{2}_{1}\,D_{a}^{2}\ln\Theta(\mathbf{Z}-\mathbf{d})+\,G_{3}-h\]
\[+\left(\kappa_{1}\,D_{a}\ln\frac{\Theta(\mathbf{Z}-\mathbf{d}+\mathbf{r})}{\Theta(\mathbf{Z}-\mathbf{d})}-G_{1}\right)^{2}+\left(\kappa_{2}\,D_{b}\ln\frac{\Theta(\mathbf{Z}-\mathbf{d}+\mathbf{r})}{\Theta(\mathbf{Z}-\mathbf{d})}+G_{2}\right)^{2}\]
\[-\kappa^{2}_{2}\,D'_{b}\ln\frac{\Theta(\mathbf{Z}-\mathbf{d}+\mathbf{r})}{\Theta(\mathbf{Z}-\mathbf{d})}+\kappa^{2}_{2}\,D_{b}^{2}\ln\frac{\Theta(\mathbf{Z}-\mathbf{d}+\mathbf{r})}{\Theta(\mathbf{Z}-\mathbf{d})}+2\,\kappa^{2}_{2}\,D_{b}^{2}\ln\Theta(\mathbf{Z}-\mathbf{d})=0. \]
By (\ref{my corol}), the last equality holds for any $\z\in\C^{g}$, and in particular for $\z=\mathbf{Z}-\mathbf{d}$.
In the same way, it can be checked that functions (\ref{sol DS}) satisfy the second equation of system (\ref{DS}).
Moreover, from (\ref{cor Fay}) we get
\[(\psi\psi^{*})_{\xi\xi}=\kappa^{3}_{1}\,\kappa_{2}\,D_{a}^{3}D_{b}\ln\,\Theta(\mathbf{Z-\mathbf{d}}), \quad (\psi\psi^{*})_{\eta\eta}=\kappa_{1}\,\kappa^{3}_{2}\,D_{a}D_{b}^{3}\ln\,\Theta(\mathbf{Z-\mathbf{d}}).\]
Therefore, taking into account that
\[\varphi_{\xi\eta}=-\frac{1}{2}\left(\kappa^{3}_{1}\,\kappa_{2}\,D_{a}^{3}D_{b}\ln\,\Theta(\mathbf{Z-\mathbf{d}})+\kappa_{1}\,\kappa^{3}_{2}\,D_{a}D_{b}^{3}\ln\,\Theta(\mathbf{Z-\mathbf{d}})\right),\]
the functions (\ref{sol DS}) satisfy the last equation of system (\ref{DS}).
\end{proof}

\vspace{0.5cm}
The solutions (\ref{sol DS}) depend on the Riemann surface $\Rs_{g}$, the points $a,b\in\Rs_{g}$, the vector $\mathbf{d}\in\C^{g}$, the constants $\kappa_{1},\kappa_{2}\in\C\setminus\{0\}$, $h\in\C$, and the local parameters $k_{a}$ and $k_{b}$ near $a$ and $b$. The transformation of the local parameters given by
\begin{align}
k_{a}&\longrightarrow\beta\,k_{a}+\mu_{1}\,k_{a}^{2}+O\left(k_{a}^{3}\right), \nonumber\\
k_{b}&\longrightarrow\beta\,k_{b}+\mu_{2}\,k_{b}^{2}+O\left(k_{b}^{3}\right), \label{trans param DS}
\end{align}
where $\beta,\mu_{1}, \mu_{2}$ are arbitrary complex numbers ($\beta\neq 0$), leads to a different family of solutions of the complexified system (\ref{DS}). These new solutions are obtained via the following transformations:
\begin{align}
\psi(\xi,\eta,t)&\longrightarrow\psi\left(\beta\,\xi+\beta\lambda_{1}\,t,\beta\,\eta+\beta\lambda_{2}\,t,\beta^{2}\,t\right)\,\exp\left\{ -\mathrm{i}\left(\lambda_{1}\,\xi+\lambda_{2}\,\eta+\left(\lambda_{1}^{2}+\lambda_{2}^{2}-\alpha\right)\,\tfrac{t}{2}\right)\right\},\\
\psi^{*}(\xi,\eta,t)&\longrightarrow\beta^{2}\,\psi^{*}\left(\beta\,\xi+\beta\lambda_{1}\,t,\beta\,\eta+\beta\lambda_{2}\,t,\beta^{2}\,t\right)\,\exp\left\{ \mathrm{i}\left(\lambda_{1}\,\xi+\lambda_{2}\,\eta+\left(\lambda_{1}^{2}+\lambda_{2}^{2}-\alpha\right)\,\tfrac{t}{2}\right)\right\}, \nonumber \\
\phi(\xi,\eta,t)&\longrightarrow\beta^{2}\,\phi\left(\beta\,\xi+\beta\lambda_{1}\,t,\beta\,\eta+\beta\lambda_{2}\,t,\beta^{2}\,t\right)+\frac{\alpha}{4}, \label{trans sol DS}
\end{align}
where $\lambda_{i}=\kappa_{i}\,\mu_{i}\,\beta^{-1}$ and $\alpha=h(1-\beta^{2})$.

\subsection{Reality condition and solutions of the DS1$^{\rho}$ equation}

Let us consider the DS1$^{\rho}$ equation
\begin{align}	
\mathrm{i}\,\psi_{t}+\frac{1}{2}(\partial_{\xi}^{2}+\partial_{\eta}^{2})\psi+2\,\phi\,\psi &=0,     \nonumber  \\
\partial_{\xi}\partial_{\eta}\,\phi+\rho\, \frac{1}{2}(\partial_{\xi}^{2}+\partial_{\eta}^{2})|\psi|^{2} & =0,  \label{DS1} 
\end{align}
where $\rho=\pm 1$. Here $\xi,\,\eta,\,t$ are real variables. Algebro-geometric solutions of (\ref{DS1}) are constructed from solutions $\psi,\,\psi^{*}$ (\ref{sol DS}) of the complexified system, under the reality condition $\psi^{*}=\rho\,\overline{\psi}$.
\\

Let $\Rs_{g}$ be a real compact Riemann surface with an anti-holomorphic involution $\tau$. Denote by $\Rs_{g}(\R)$ the set of fixed points of the involution $\tau$ (see Appendix A.1). Let us choose the homology basis satisfying (\ref{hom basis}). Then the solutions of (\ref{DS1}) are given by

\begin{theorem}
Let $a,b\in\Rs_{g}(\R)$ be distinct points with local parameters satisfying $\overline{k_{a}(\tau p)}=k_{a}(p)$ for any $p$ lying in a neighbourhood of $a$, and $\overline{k_{b}(\tau p)}=k_{b}(p)$ for any $p$ lying in a neighbourhood of $b$. Denote by $\{\mathcal{A},\mathcal{B},\ell\}$ the standard generators of the relative homology group $H_{1}(\Rs_g,\{a,b\})$ (see Appendix A.2). 
Let $\mathbf{d}_{R}\in\R^{g}$, $\mathbf{T}\in\Z^{g}$, and define $\mathbf{d}=\mathbf{d}_{R}+\frac{\mathrm{i}\pi}{2}(\text{diag}(\mathbb{H})-2\,\mathbf{T})$. Morover,
take $\theta,\,h,\in\R$, $\tilde{\kappa}_{1},\,\kappa_{2}\in\R\setminus\left\{0\right\}$ and put 
\begin{equation}
\kappa_{1}=-\rho\,\tilde{\kappa}_{1}^{2}\,\kappa_{2}\,q_{2}(a,b)\,\exp\left\{\tfrac{1}{2}\left\langle \B \mathbf{M},\mathbf{M}\right\rangle+\left\langle \mathbf{r}+\mathbf{d},\mathbf{M}\right\rangle\right\}, \label{kappa DS1}
\end{equation}
where $\mathbf{M}\in\Z^{g}$ is defined in (\ref{hom basis 3}).
Then the following functions $\psi$ and $\phi$ are solutions of the DS1$^{\rho}$ equation 
\begin{equation}
\psi(\xi,\eta,t)=|A|\,e^{\mathrm{i}\theta}\,\frac{\Theta(\mathbf{Z}-\mathbf{d}+\mathbf{r})}{\Theta(\mathbf{Z}-\mathbf{d})}\,\exp\left\{ \mathrm{i}\left(-G_{1}\,\xi-G_{2}\,\eta+G_{3}\,\tfrac{t}{2}\right)\right\}, \label{psi DS1}
\end{equation}
\begin{equation}
\phi(\xi,\eta,t)=\frac{1}{2}\,(\ln\,\Theta(\mathbf{Z}-\mathbf{d}))_{\xi\xi}+\frac{1}{2}\,(\ln\,\Theta(\mathbf{Z}-\mathbf{d}))_{\eta\eta}+\frac{1}{4}h,\label{sol DS1}
\end{equation}
 where $|A|=\left|\tilde{\kappa}_{1}\,\kappa_{2}\,q_{2}(a,b)\right|\,\exp\left\{\left\langle \mathbf{d}_{R},\mathbf{M}\right\rangle\right\}.$
Here $\mathbf{r}=\int_{\ell}\omega$, and the vector $\mathbf{Z}$ is defined in (\ref{Z DS}). Scalars $q_{2}(a,b),\,G_{1},\,G_{2}$ and $G_{3}$ are defined in (\ref{q2}), (\ref{N12 DS}) and (\ref{N3 DS}) respectively.
\end{theorem}
\vspace{0.5cm}
\noindent
The case where $\mathbf{V}_{a}+\mathbf{V}_{b}=0$ and $\kappa_{1}=\kappa_{2}$ is treated at the end of this section. It corresponds to solutions of the nonlinear Schrödinger equation.

\begin{proof}
Let us check that under the conditions of the theorem, the functions $\psi$ and $\psi^{*}$ (\ref{sol DS}) satisfy the reality conditions (\ref{real cond DS}). First of all, invariance with respect to the anti-involution $\tau$ of the points $a$ and $b$ implies the reality of vector (\ref{Z DS}):
\begin{equation}
\overline{\mathbf{Z}}=\mathbf{Z}. \label{vectZ DS1}
\end{equation}
In fact, using the expansion (\ref{exp hol diff}) of the normalized holomorphic differentials $\omega_{j}$ near $a$ we get
\[\overline{\tau^{*} \omega_{j}}(a)(p)=(\overline{V_{a,j}}+\overline{W_{a,j}}\,k_{a}(p)+\ldots)\,\ud k_{a}(p),\]
for any point $p$ lying in a neighbourhood of $a$.
Then by (\ref{diff hol}), the vectors $\mathbf{V}_{a}$ and $\mathbf{W}_{a}$ appearing in expression (\ref{Z DS}) satisfy
\begin{equation}
\overline{\mathbf{V}_{a}}=-\mathbf{V}_{a}, \quad \overline{\mathbf{W}_{a}}=-\mathbf{W}_{a}. \label{vectV DS1}
\end{equation}
The same holds for the vectors $\mathbf{V}_{b}$ and $\mathbf{W}_{b}$, which leads to (\ref{vectZ DS1}). 
Moreover, from (\ref{diff hol}) and (\ref{hom basis 3}) we get
\begin{equation}
\overline{\mathbf{r}}=-\mathbf{r}-2\mathrm{i}\pi\mathbf{N}-\mathbb{B}\mathbf{M}, \label{r DS1}
\end{equation}
where $\mathbf{N},\,\mathbf{M}\in\Z^{g}$ are defined in (\ref{hom basis 3}) and satisfy 
\begin{equation}
2\,\mathbf{N}+\mathbb{H}\mathbf{M}=0. \label{NHM}
\end{equation}
From (\ref{my corol}), it is straightforward to see that the scalars $K_{1}(a,b)$ and $K_{2}(a,b)$ defined by (\ref{K1}) and (\ref{K2}) satisfy
\begin{equation}
\overline{K_{1}(a,b)}=K_{1}(a,b)-\left\langle \mathbf{V}_{a},\mathbf{M}\right\rangle, \qquad
\overline{K_{2}(a,b)}=K_{2}(a,b)+\left\langle \mathbf{W}_{a},\mathbf{M}\right\rangle, \label{K1 K2 conj}
\end{equation}
which implies
\begin{equation}
\overline{G_{1}}=G_{1}-\kappa_{1}\left\langle \mathbf{V}_{a},\mathbf{M}\right\rangle, \qquad
\overline{G_{2}}=G_{2}-\kappa_{2}\left\langle \mathbf{V}_{b},\mathbf{M}\right\rangle, \qquad \overline{G_{3}}=G_{3}+\kappa_{1}^{2}\left\langle \mathbf{W}_{a},\mathbf{M}\right\rangle+\kappa_{2}^{2}\left\langle \mathbf{W}_{b},\mathbf{M}\right\rangle. \nonumber
\end{equation}
Therefore, the reality condition (\ref{real cond DS}) together with (\ref{sol DS}) leads to 
\begin{multline}
|A|^{2}=-\rho\,\kappa_{1}\kappa_{2}\,q_{2}(a,b)\,\frac{\Theta(\mathbf{Z}-\mathbf{d}-\mathbf{r})\,\Theta(\mathbf{\mathbf{Z}-\overline{\mathbf{d}}}+\mathrm{i}\pi \,\text{diag}(\mathbb{H}))}{\Theta(\mathbf{Z}-\overline{\mathbf{d}}-\mathbf{r}+\mathrm{i}\pi \,\text{diag}(\mathbb{H}))\,\Theta(\mathbf{\mathbf{Z}-\mathbf{d}})}\\
\times\exp\left\{\tfrac{1}{2}\left\langle \B \mathbf{M},\mathbf{M}\right\rangle+\left\langle \mathbf{r}+\overline{\mathbf{d}}-\mathrm{i}\pi\text{diag}(\mathbb{H}),\mathbf{M}\right\rangle\right\},\label{4.7}
\end{multline}
taking into account the action (\ref{conj theta}) of the complex conjugation on the theta function, and the quasi-periodicity (\ref{2.4}) of the theta function. Let us choose a vector $\mathbf{d}\in \C^{g}$ such that
\[\overline{\mathbf{d}}\equiv\mathbf{d}- \mathrm{i}\pi \,\text{diag}(\mathbb{H}) \,\,\, \text{mod} \,\, (2\mathrm{i}\pi \Z^{g}+ \mathbb{B}\,\Z^{g}), \]
which is, since $\overline{\mathbf{d}}-\mathbf{d}$ is purely imaginary, equivalent to $\overline{\mathbf{d}}= \mathbf{d}- \mathrm{i}\pi \,\text{diag}(\mathbb{H}) + 2i\pi\mathbf{T},$ for some $\mathbf{T}\in\Z^{g}$. 
Here we used the action (\ref{matrix B}) of the complex conjugation on the matrix of $\mathcal{B}$-periods $\B$, and the fact that $\B$ has a negative definite real part.
 Hence, the vector $\mathbf{d}$ can be written as 
\begin{equation}
 \mathbf{d}=\mathbf{d}_{R}+\frac{\mathrm{i}\pi}{2}(\text{diag}(\mathbb{H})-2\,\mathbf{T}), \label{vectD DS1}
\end{equation}
 for some $\mathbf{d}_{R}\in\R^{g}$ and $\mathbf{T}\in\Z^{g}$.
Therefore, all theta functions in (\ref{4.7}) cancel out and (\ref{4.7}) becomes
\begin{equation}
|A|^{2}=-\rho\,\kappa_{1}\kappa_{2}\,q_{2}(a,b)\,\exp\left\{\tfrac{1}{2}\left\langle \B \mathbf{M},\mathbf{M}\right\rangle+\left\langle \mathbf{r}+\mathbf{d},\mathbf{M}\right\rangle\right\}. \label{4.9}
\end{equation}
The reality of the right hand side of equality (\ref{4.9}) can be deduced from formula (\ref{arg q2}) for the argument of $q_{2}(a,b)$.
Moreover, it is straightforward to see from (\ref{vectD DS1}) and (\ref{NHM}) that $\exp\{\left\langle \mathbf{d},\mathbf{M}\right\rangle\}$ is also real. Since $\kappa_{1},\,\kappa_{2}$ are arbitrary real constants, we can choose $\kappa_{1}$ as in (\ref{kappa DS1}), which leads to
\[|A|^{2}=\left(\tilde{\kappa}_{1}\,\kappa_{2}\,q_{2}(a,b)\,\exp\left\{\tfrac{1}{2}\left\langle \B \mathbf{M},\mathbf{M}\right\rangle+\left\langle \mathbf{r}+\mathbf{d},\mathbf{M}\right\rangle\right\}\right)^{2}
=\left|\tilde{\kappa}_{1}\,\kappa_{2}\,q_{2}(a,b)\right|^{2}\,\exp\left\{2\left\langle \mathbf{d}_{R},\mathbf{M}\right\rangle\right\}.\]
\end{proof}
\vspace{0.5cm}

Functions $\psi$ and $\phi$ given in (\ref{psi DS1}) and (\ref{sol DS1}) describe a family of algebro-geometric solutions of (\ref{DS1}) depending on: a real Riemann surface $(\Rs_{g},\tau)$, two distinct points $a,b\in\Rs_{g}(\R)$, local parameters $k_{a},k_{b}$ which satisfy $\overline{k_{a}(\tau p)}=k_{a}(p)$ and $\overline{k_{b}(\tau p)}=k_{b}(p)$, and arbitrary constants $\mathbf{d}_{R}\in\R^{g}$, $\mathbf{T}\in\Z^{g}$, $\theta, h,\in\R$, $\tilde{\kappa}_{1},\,\kappa_{2}\in\R\setminus\left\{0\right\}$. Note that by periodicity properties of the theta function, without loss of generality, the vector $\mathbf{T}$ can be chosen in the set $\{0,1\}^{g}$. The case where the Riemann surface is dividing and $\mathbf{T}=0$ is of special importance, because the related solutions are smooth, as explained in the next proposition.
\\

Since the theta function is entire, singularities of the functions $\psi$ and $\phi$ can appear only at the zeros of their denominator. Following Vinnikov's result \cite{Vin} we obtain

\begin{proposition}
Solutions (\ref{psi DS1}) and (\ref{sol DS1}) are 
smooth if the curve $\Rs_{g}$ is dividing and $\mathbf{d}\in \R^{g}$. Assume that solutions (\ref{psi DS1}) and (\ref{sol DS1}) are smooth for any vector  $\mathbf{d}$ lying in a component $T_{v}$ (\ref{3.10}) of the Jacobian, then the curve is dividing and $\mathbf{d}\in \R^{g}$.
\end{proposition}

\begin{proof}
By (\ref{vectZ DS1}) and (\ref{vectD DS1}), the vector $\mathbf{Z}-\mathbf{d}$ belongs 
to the set $S_{1}$ introduced in (\ref{S1}). Hence by Proposition A.3, the solutions are smooth if the curve is 
dividing (in this case $\text{diag}(\mathbb{H})$=0), and if the argument $\mathbf{Z}-\mathbf{d}$ of the theta function in the denominator is real, which by (\ref{vectZ DS1}) leads to the choice $\mathbf{d}\in\R^{g}$ (and then $\mathbf{T}=0$ in Theorem 3.2). 

The following assertions were proved in \cite{Vin}: let $\Rs_{g}(\R)\neq\emptyset$; if $\Rs_{g}$ is non dividing, then $T_{v}\cap (\Theta)\neq \emptyset$ for all $v$, where $(\Theta)$ denotes the set of zeros of the theta function; if $\Rs_{g}$ is dividing, then $T_{v}\cap \Theta\neq\emptyset$ if and only if $v\neq 0$. It follows that if solutions are smooth for any vector  $\mathbf{d}$ lying in a component $T_{v}$ (\ref{3.10}) of the Jacobian, then the curve is dividing and $v=0$. Hence $\mathbf{d}\in T_{0}$ where $T_{0}=\R^{g}$.
\end{proof}

\subsection{Reality condition and solutions of the DS2$^{\rho}$ equation}
Let us consider the DS2$^{\rho}$ equation
\begin{align}	
\mathrm{i}\,\psi_{t}+\frac{1}{2}(\partial_{\xi}^{2}+\partial_{\eta}^{2})\psi+2\,\phi\,\psi &=0,     \nonumber  \\
\partial_{\xi}\partial_{\eta}\,\phi+\rho\, \frac{1}{2}(\partial_{\xi}^{2}+\partial_{\eta}^{2})|\psi|^{2} & =0,  \label{DS2} 
\end{align}
where $\rho=\pm 1$. Here $t$ is a real variable and variables $\xi,\,\eta$ satisfy $\overline{\xi}=\eta$. Analogously to the case where $\xi$ and $\eta$ are real variables (see Section 3.2), algebro-geometric solutions of (\ref{DS2}) are constructed from solutions $\psi,\,\psi^{*}$ (\ref{sol DS}) of the complexified system by imposing the reality condition $\psi^{*}=\rho\,\overline{\psi}$.
\\

Let $\Rs_{g}$ be a real compact Riemann surface with an anti-holomorphic involution $\tau$. Let us choose the homology basis satisfying (\ref{hom basis}). Then the solutions of (\ref{DS2}) are given by

\begin{theorem}
Let $a,b\in\Rs_{g}$ be distinct points such that $\tau a=b$, with local parameters satisfying $\overline{k_{b}(\tau p)}=k_{a}(p)$ for any point $p$ lying in a neighbourhood of $a$. Denote by $\{\mathcal{A},\mathcal{B},\ell\}$ the standard generators of the relative homology group $H_{1}(\Rs_{g},\{a,b\})$ (see Appendix A.2). Let $\mathbf{T},\,\mathbf{L}\in\Z^{g}$ satisfy
\begin{equation}
2\,\mathbf{T}+ \mathbb{H}\,\mathbf{L}=\text{diag}(\mathbb{H}), \label{TL DS2}
\end{equation}
and define $\mathbf{d}=\frac{1}{2}\mathrm{Re}(\mathbb{B})\,\mathbf{L}+\mathrm{i}\mathbf{d}_{I},$ for some $\mathbf{d}_{I}\in\R^{g}$.
Moreover, take $\theta,\,h\in\R$ and $\kappa_{1},\kappa_{2}\in\C\setminus\left\{0\right\}$ such that $\overline{\kappa_{1}}=\kappa_{2}$. Let us consider the following functions $\psi$ and $\phi$:
\begin{equation}
\psi(\xi,\eta,t)=|A|\,e^{\mathrm{i}\theta}\,\frac{\Theta(\mathbf{Z}-\mathbf{d}+\mathbf{r})}{\Theta(\mathbf{Z}-\mathbf{d})}\,\exp\left\{ \mathrm{i}\left(-G_{1}\,\xi-G_{2}\,\eta+G_{3}\,\tfrac{t}{2}\right)\right\},  \label{psi DS2}
\end{equation}
\begin{equation}
\phi(\xi,\eta,t)=\frac{1}{2}\,(\ln\,\Theta(\mathbf{Z}-\mathbf{d}))_{\xi\xi}+\frac{1}{2}\,(\ln\,\Theta(\mathbf{Z}-\mathbf{d}))_{\eta\eta}+\frac{1}{4}h,\label{sol DS2}
\end{equation}
\vspace{0.3cm}
where $|A|=|\kappa_{1}|\,|q_{2}(a,b)|^{1/2}\,\exp\left\{-\frac{1}{2}\left\langle \mathrm{Re}(\mathbf{r}),\mathbf{L}\right\rangle\right\}.$
Then,
\begin{enumerate}
\item if $\ell$ intersects the set of real ovals of $\Rs_{g}$ only once, and if this intersection is transversal, functions $\psi$ and $\phi$ are solutions of DS2$^{\rho}$ whith $\rho=e^{\mathrm{i}\pi\left\langle \mathbf{N},\mathbf{L}\right\rangle}$,
\item if $\ell$ does not cross any real oval, functions $\psi$ and $\phi$ are solutions of DS2$^{\rho}$ whith $\rho=-e^{\mathrm{i}\pi\left\langle \mathbf{N},\mathbf{L}\right\rangle}$.
\end{enumerate}
Here $\mathbf{r}=\int_{\ell}\omega$, the vector $\mathbf{Z}$ is defined in (\ref{Z DS}) and vector $\mathbf{N}\in\Z^{g}$ is defined in (\ref{hom basis 1}). Scalars $q_{2}(a,b),\,G_{1},\,G_{2}$ and $G_{3}$ are defined in (\ref{q2}), (\ref{N12 DS}) and (\ref{N3 DS}) respectively.
\end{theorem}

\begin{proof}
Analogously to the proof of Theorem 3.2, let us check that under the conditions of the theorem, the functions $\psi^{*}$ and $\psi$ (\ref{DS}) satisfy the reality condition (\ref{real cond DS}).
First of all, due to the fact that points $a$ and $b$ are interchanged by $\tau$, the vector $\mathbf{Z}$ (\ref{Z DS}) satisfies
\begin{equation}
\overline{\mathbf{Z}}= -\mathbf{Z}. \label{vectZ DS2}
\end{equation}
In fact, using the expansion (\ref{exp hol diff}) of the normalized holomorphic differentials $\omega_{j}$ near $a$ we get
\[\overline{\tau^{*} \omega_{j}}(a)(p)=(\overline{V_{b,j}}+\overline{W_{b,j}}\,k_{a}(p)+\ldots)\,\ud k_{a}(p),\]
for any point $p$ lying in a neighbourhood of $a$.
Then by (\ref{diff hol}) the vectors $\mathbf{V}_{a},\,\mathbf{V}_{b}$ and $\mathbf{W}_{a},\,\mathbf{W}_{b}$ appearing in the vector $\mathbf{Z}$ satisfy
\begin{equation}
\overline{\mathbf{V}_{a}}=-\mathbf{V}_{b}, \quad \overline{\mathbf{W}_{a}}=-\mathbf{W}_{b}, \label{vectV DS2}
\end{equation}
which leads to (\ref{vectZ DS2}). 
From (\ref{diff hol}) and (\ref{hom basis 1}) we get 
\begin{equation}
\overline{\mathbf{r}}=\mathbf{r}- 2\mathrm{i}\pi\, \mathbf{N}, \label{r DS2}
\end{equation}
where $\mathbf{N}\in\Z^{g}$ is defined in (\ref{hom basis 1}).
By Proposition B.3, the scalar $q_{2}(a,b)$ is real. From (\ref{my corol}), it is straightforward to see that the scalars $K_{1}(a,b)$ and $K_{2}(a,b)$, defined in (\ref{K1}) and (\ref{K2}), satisfy
\[\overline{K_{1}(a,b)}=K_{1}(b,a), \qquad \overline{K_{2}(a,b)}=K_{2}(b,a),\]
which leads to $\overline{G_{1}}=G_{2}$ and $G_{3}\in\R$. 
Therefore, the reality condition (\ref{real cond DS}) together with (\ref{sol DS}) leads to 
\begin{equation}
|A|^{2}=-\rho\,|\kappa_{1}|^{2}\,q_{2}(a,b)\frac{\Theta(\mathbf{Z}-\mathbf{d}-\mathbf{r})\,\Theta(\mathbf{\mathbf{Z}+\overline{\mathbf{d}}}+\mathrm{i}\pi \,\text{diag}(\mathbb{H}))}{\Theta(\mathbf{Z}+\overline{\mathbf{d}}-\mathbf{r}+\mathrm{i}\pi \,\text{diag}(\mathbb{H}))\,\Theta(\mathbf{\mathbf{Z}-\mathbf{d}})},\label{4.12}
\end{equation}
taking into account (\ref{conj theta}).
Let us choose a vector $\mathbf{d}\in \C^{g}$ such that
\begin{equation}
 \overline{\mathbf{d}}= -\mathbf{d}- \mathrm{i}\pi \,\text{diag}(\mathbb{H}) + 2\mathrm{i}\pi\mathbf{T}+\B\,\mathbf{L},  \nonumber
\end{equation}
for some vector $\mathbf{T},\,\mathbf{L}\in\Z^{g}$. The reality of the vector  $\overline{\mathbf{d}}+\mathbf{d}$ together with (\ref{matrix B}) imply 
\begin{equation}
\mathbf{d}=\frac{1}{2}\mathrm{Re}(\mathbb{B})\,\mathbf{L}+\mathrm{i}\mathbf{d}_{I}\label{4.13}
\end{equation}
for some $\mathbf{d}_{I}\in\R^{g}$, where $2\,\mathbf{T}+ \mathbb{H}\,\mathbf{L}=\text{diag}(\mathbb{H})$.
With this choice of vector $\mathbf{d}$, (\ref{4.12}) becomes
\begin{equation}
|A|^{2}=-\rho\,|\kappa_{1}|^{2}\,q_{2}(a,b)\,e^{-\left\langle \mathbf{r},\mathbf{L}\right\rangle}.  \label{4.14}
\end{equation}
Moreover, from (\ref{r DS2}) we deduce that equality (\ref{4.14}) holds only if 
\[\rho=-\text{sign}(q_{2}(a,b))\,e^{-\mathrm{i}\pi\left\langle \mathbf{N},\mathbf{L}\right\rangle}.\]
The sign of $q_{2}(a,b)$ in the case where $\tau a =b$ is given in Proposition B2, which completes the proof.
\end{proof}

\vspace{0.3cm}
\noindent
\begin{corollary} From Theorem 3.3 we deduce that
\begin{enumerate}
\item if $\Rs_{g}$ is dividing and each component of $\mathbf{L}$ is even, functions (\ref{psi DS2}) and (\ref{sol DS2}) are solutions of DS2$^{+}$,
\item if $\Rs_{g}$ does not have real ovals and each component of $\mathbf{L}$ is even, functions (\ref{psi DS2}) and (\ref{sol DS2}) are solutions of DS2$^{-}$.
\end{enumerate}
\end{corollary}

\noindent
\begin{remark}
\rm{To construct solutions associated to non-dividing Riemann surfaces, we first observe from (\ref{TL DS2}) that all components of the vector $\mathbf{L}$ cannot be even, since for non dividing Riemann surfaces the vector $\text{diag}(\mathbb{H})$ contains odd coefficients (see Appendix A.1). In this case, the vector $\mathbf{N}$ has to be computed to determine the sign $\rho=-e^{\mathrm{i}\pi\left\langle \mathbf{N},\mathbf{L}\right\rangle}$ in the reality condition. This vector $\mathbf{N}$ is defined by the action of $\tau$ on the relative homology group $H_{1}(\Rs_{g},\{a,b\})$ (see (\ref{hom basis 1})). It follows that we do not have a general expression for this vector.}
\end{remark}

To ensure the smoothness of solutions (\ref{psi DS2}) and (\ref{sol DS2}) for all complex 
conjugate $\xi, \,\eta,$ and $t\in\R$, the function 
$\Theta(\mathbf{Z}-\mathbf{d})$ of the variables $\xi,\,\eta,\,t$ must not vanish. Following the work by Dubrovin and Natanzon \cite{DN} on smoothness of algebro-geometric solutions of the Kadomtsev Petviashvili (KP1) equation in the case where $\Rs_{g}$ admits real ovals we get

\begin{proposition} Functions (\ref{psi DS2}) and (\ref{sol DS2}) are smooth solutions of DS2$^{+}$ if the curve is an M-curve and $\mathbf{d}\in \mathrm{i}\,\R^{g}$. Assume that the curve admits real ovals and functions (\ref{psi DS2}), (\ref{sol DS2}) are smooth solutions of DS2$^{\rho}$ for any vector $\mathbf{d}$ lying in a component $\tilde{T}_{v}$ (\ref{3.12}) of the Jacobian, then the curve is an M-curve, $\mathbf{d}\in \mathrm{i}\,\R^{g}$ and $\rho=+1$.
\end{proposition}

\begin{proof}
By 
(\ref{vectZ DS2}) and (\ref{4.13}) the vector $\mathbf{Z}-\mathbf{d}$ 
belongs to the set $S_{2}$ introduced in (\ref{S2}). Hence by Proposition A.4, the solutions are smooth if the curve is an M-curve and 
$\mathbf{Z}-\mathbf{d}\in \mathrm{i}\,\R^{g}$ which implies $\mathbf{d}\in \mathrm{i}\,\R^{g}$ by (\ref{vectZ DS2}) (and therefore $\mathbf{L}=\mathbf{T}=0$). 
\end{proof}

\noindent
\begin{remark} 
\rm{Smoothness of solutions of the DS2$^{-}$ equation was investigated in \cite{Mal}. It is proved that solutions are smooth if and only if the associated Riemann surface does not have real ovals, and if there are no pseudo-real functions of degree $g-1$ on it (i.e. functions which satisfy $\overline{f(\tau p)}=-f(p)^{-1}$).}
\end{remark}

\subsection{Reduction of the DS1$^{\rho}$ equation to the NLS equation}

Solutions of the nonlinear Schrödinger equation (\ref{NLS}) can be derived from solutions of the Davey Stewartson equations, when the associated Riemann surface is hyperelliptic. 

\begin{proposition}
Let $\Rs_{g}$ be a hyperelliptic curve of genus $g$ which admits an anti-holomorphic involution $\tau$. Denote by $\sigma$ the hyperelliptic involution defined on $\Rs_{g}$. Let $a,b\in\Rs_{g}(\R)$ with local parameters satisfying $\overline{k_{a}(\tau p)}=k_{a}(p)$ for $p$ near $a$, and $\overline{k_{b}(\tau p)}=k_{b}(p)$ for $p$ near $b$. Moreover, assume that $\sigma a =b$ and $k_{a}(p)=k_{b}(\sigma p)$. Then, taking $\kappa_{1}=\kappa_{2}=1$, the function $\psi$ in (\ref{psi DS1}) is a solution of the equation  
\[\mathrm{i}\,\psi_{t}+\psi_{\xi\xi}+2\,\left(\rho\,|\psi|^{2}+q_{1}(a,b)+\tfrac{1}{4}h\right)\,\psi=0,\]
which can be transformed to the NLS$^{\rho}$ equation 
\[\mathrm{i}\,\tilde{\psi}_{t}+\tilde{\psi}_{\xi\xi}+2\,\rho\,|\tilde{\psi}|^{2}\,\tilde{\psi}=0,\]
by the substitution
\[\tilde{\psi}(\xi,t)= \psi(\xi,t)\,\exp\left\{-2\mathrm{i}\left(q_{1}(a,b)+\tfrac{1}{4}h\right)\,t\right\}.\]
If all branch points of $\Rs_{g}$ are real, $\tilde{\psi}$ is a smooth solution of NLS$^{-}$. If they are all pairwise conjugate, $\tilde{\psi}$ is a smooth solution of NLS$^{+}$.
\end{proposition}

\begin{proof}
If $a,b\in\Rs_{g}$ are such that $\sigma a =b$ and the local parameters satisfy $k_{a}(p)=k_{b}(\sigma p)$, one has 
\begin{equation}
\mathbf{V}_{a}+\mathbf{V}_{b}=0 ,\qquad \mathbf{W}_{a}+\mathbf{W}_{b}=0. \label{V W NLS}
\end{equation} 
To verify (\ref{V W NLS}), we use the action $\sigma \mathcal{A}_{k} = -\mathcal{A}_{k}$ of the involution $\sigma$ on the $\mathcal{A}$-cycles of the homology basis. Hence by (\ref{norm hol diff}) we have
\[2\mathrm{i}\pi\delta_{jk} = \int_{\sigma \mathcal{A}_{k}}\sigma^{*}\omega_{j}=-\int_{\mathcal{A}_{k}}\sigma^{*}\omega_{j}.\]
It follows that the holomorphic differential $-\sigma^{*}\omega_{j}$ satisfies the normalization condition (\ref{norm hol diff}), which implies, by virtue of uniqueness of the normalized holomorphic differentials, 
\[\sigma^{*}\omega_{j}=-\omega_{j}.\]
Using (\ref{exp hol diff}) we obtain
\begin{align}
\sigma^{*} \omega_{j}(a)(p)&=(V_{b,j}+W_{b,j}\,k_{b}(\sigma p)+\ldots)\,\ud k_{b}(\sigma p)\nonumber\\
&=(V_{b,j}+W_{b,j}\,k_{a}(p)+\ldots)\,\ud k_{a}(p),\nonumber
\end{align}
which implies $\mathbf{V}_{a}+\mathbf{V}_{b}=0$ and 
$\mathbf{W}_{a}+\mathbf{W}_{b}=0$.

Therefore, when the Riemann surface associated to solutions of DS1$^{\rho}$ is hyperelliptic, assuming that $a$ and $b$ satisfy $\sigma a=b$, and $\kappa_{1}=\kappa_{2}=1$, by (\ref{V W NLS}) and (\ref{cor Fay}), under the reality condition $\psi^{*}=\rho\, \overline{\psi}$, the function $\phi$ in (\ref{sol DS1}) satisfies \[\phi(\xi,\eta,t)=\rho\,|\psi|^{2}+q_{1}(a,b)+\frac{1}{4}h.\]
Hence the function $\psi$ (\ref{psi DS1}) becomes a solution of the equation  
\[\mathrm{i}\,\psi_{t}+\psi_{\xi\xi}+2\,\left(\rho\,|\psi|^{2}+q_{1}(a,b)+\tfrac{1}{4}h\right)\,\psi=0,\]
with $\rho=\pm 1$, depending on the reality of the branch points as explained in Section 4.
\end{proof}

\vspace{0.3cm}
\noindent
Solutions of the NLS equation obtained in this way coincide with those in \cite{BBEIM}.

\section{Algebro-geometric solutions of the multi-component NLS equation}

In this section, we present another application of the degenerated Fay identity (\ref{my corol}), which 
leads to new theta-functional solutions of the multi-component nonlinear Schr\"odinger equation (n-NLS$^{s}$)
\begin{equation}		
\mathrm{i}\,\frac{\partial \psi_{j}}{\partial t}+\frac{\partial^{2} \psi_{j}}{\partial x^{2}}+2\,\left(\sum_{k=1}^{n}s_{k}|\psi_{k}|^{2}\right)\,\psi_{j} =0,  \quad \quad j=1,\ldots,n,   \label{n-NLS bis}
\end{equation}
where $s=(s_{1},\ldots,s_{n})$,  $s_{i}=\pm 1$. Here $\psi_{j}(x,t)$ are complex valued functions of the real variables $x$ and $t$.

\subsection{Solutions of the complexified n-NLS equation}

Consider first the complexified version of the n-NLS$^{s}$ equation, which is a system of $2n$ equations of $2n$ dependent variables $\left\{\psi_{j},\,\psi_{j}^{*}\right\}_{j=1}^{n}$
\begin{align}		
\mathrm{i}\,\frac{\partial \psi_{j}}{\partial t}+\frac{\partial^{2} \psi_{j}}{\partial x^{2}}+2\,\left(\sum_{k=1}^{n}\,\psi_{k}\,\psi_{k}^{*}\right)\,\psi_{j} &=0  ,  \nonumber  \\ 
-\mathrm{i}\,\frac{\partial \psi_{j}^{*}}{\partial t}+\frac{\partial^{2} \psi_{j}^{*}}{\partial x^{2}}+2\,\left(\sum_{k=1}^{n}\,\psi_{k}\,\psi_{k}^{*}\right)\,\psi_{j}^{*} &=0 , \quad \quad j=1,\ldots,n  ,\label{ass n-NLS} 
\end{align}
where $\psi_{j}(x,t)$ and $\psi_{j}^{*}(x,t)$ are complex valued functions of the real variables $x$ and $t$. This system reduces to the n-NLS$^{s}$ equation (\ref{n-NLS bis}) under the \textit{reality conditions }
\begin{equation}
\psi_{j}^{*}=s_{j}\,\overline{\psi_{j}}, \quad\quad j=1,\ldots, n. \label{real cond n-NLS}
\end{equation}
Theta functional solutions of the system (\ref{ass n-NLS}) are given by

\begin{theorem}
Let $\mathcal{R}_{g}$ be a compact Riemann surface of genus $g>0$ and let $f$ be a meromorphic function
of degree $n+1$ on $\Rs_{g}$. Let $z_{a}\in\C$ be a non critical value of $f$, and consider the fiber $f^{-1}(z_{a})=\left\{a_{1},\ldots,a_{n+1}\right\}$ over $z_{a}$. Choose the local parameters $k_{a_{j}}$ near $a_{j}$ as $k_{a_{j}}(p)=f(p)-z_{a}$, 
for any point $p\in\Rs_{g}$ lying in a neighbourhood of $a_{j}$.
Let $\mathbf{d}\in\C^{g}$ and $A_{j}\neq 0$ be arbitrary constants. Then the following functions  $\left\{\psi_{j}\right\}_{j=1}^{n}$ and  $\left\{\psi_{j}^{*}\right\}_{j=1}^{n}$ are solutions of the system (\ref{ass n-NLS})
\begin{align}    
\psi_{j}(x,t)&=A_{j}\,\frac{\Theta(\mathbf{Z}-\mathbf{d}+\mathbf{r}_{j})}{\Theta(\mathbf{Z}-\mathbf{d})}\,\exp\left\{ \mathrm{i}(-E_{j}\,x+\,F_{j}\,t)\right\},\nonumber\\ 
\psi^{*}_{j}(x,t)&=\frac{q_{2}(a_{n+1},a_{j})}{A_{j}}\,\frac{\Theta(\mathbf{Z}-\mathbf{d}-\mathbf{r}_{j})}{\Theta(\mathbf{Z}-\mathbf{d})}\,\exp\left\{ \mathrm{i}(E_{j}\,x-\,F_{j}\,t)\right\}. \label{sol n-NLS comp}
\end{align}
Here $\mathbf{r}_{j}=\int^{a_{j}}_{a_{n+1}}\omega$, where $\omega$ is the vector of normalized holomorphic differentials, and
\begin{equation}
\mathbf{Z}=\mathrm{i}\,\mathbf{V}_{a_{n+1}}\,x+\,\mathrm{i}\,\mathbf{W}_{a_{n+1}}\,t. \label{Z n-NLS}
\end{equation}
The vectors $\mathbf{V}_{a_{n+1}}$ and $\mathbf{W}_{a_{n+1}}$ are defined in (\ref{exp hol diff}), and the scalars $E_{j},\,F_{j}$ are given by
\begin{equation}
E_{j}=K_{1}(a_{n+1},a_{j}),\qquad F_{j}=K_{2}(a_{n+1},a_{j})-2\,\sum_{k=1}^{n}q_{1}(a_{n+1},a_{k}). \label{E,N n-NLS}
\end{equation}
The scalars $q_2(a_{n+1},a_{j}), \,K_{1}(a_{n+1},a_{j}),\,K_{2}(a_{n+1},a_{j})$ and $q_{1}(a_{n+1},a_{k})$ are defined in (\ref{q2}), (\ref{K1}), (\ref{K2}) and (\ref{q1}) respectively. 
\end{theorem}

\begin{proof}
We start with the following technical lemma.

\begin{lemma}
Let $\mathcal{R}_{g}$ be a compact Riemann surface of genus $g>0$ and let $a_{1},\ldots,a_{n+1}$ be distinct points on $\Rs_{g}$. Then the vectors $\mathbf{V}_{a_{j}}$ for $j=1,\ldots,n+1$ are linearly dependent if and only if there exists a meromorphic function $f$ of degree $n+1$ on $\Rs_{g}$, and $z_{a}\in\C\mathbb{P}^{1}$ such that $f^{-1}(z_{a})=\left\{a_{1},\ldots,a_{n+1}\right\}$.
\end{lemma}
\vspace{0.3cm}
\noindent
\textit{Proof of Lemma 4.1.}
Assume that there exist $\alpha_{1},\ldots,\alpha_{n+1}\in\C^{*}$ such that $\sum_{k=1}^{n+1}\alpha_{k}\,\mathbf{V}_{a_{k}}=0$. The left hand side of this equality equals the vector of $\mathcal{B}$-periods (see e.g. \cite{Bob}) of the normalized differential of the second kind $\Omega=\sum_{k=1}^{n+1}\alpha_{k}\, \Omega_{a_{k}}^{(2)}$. Hence all periods of the differential $\Omega$ vanish, which implies that the Abelian integral $p\longmapsto\int^{p}_{p_{0}}\Omega$ is a meromorphic function of degree $n+1$ on $\Rs_{g}$ having simple poles at $a_{1},\ldots,a_{n+1}$. 

Conversely, assume that there exists a meromorphic function $f$ of degree $n+1$ on $\Rs_{g}$, and $z_{a}\in\C$ such that $f^{-1}(z_{a})=\left\{a_{1},\ldots,a_{n+1}\right\}$ (the case $z_{a}=\infty$ can be treated in the same way).
The function $h(p)=(f(p)-z_{a})^{-1}$ is a meromorphic function of degree $n+1$ on $\Rs_{g}$ having simple poles at $a_{1},\ldots,a_{n+1}$ only. Therefore all periods of the differential $\mathrm{d}h$ vanish. Let $p_{0}\in\Rs_{g}$ satisfy $h(p_{0})=0$. Using Riemann's bilinear identity \cite{Bob} we get
\[\int_{\partial F_{g}}\omega_{j}\int^{p}_{p_{0}}\mathrm{d}h=\int_{\partial F_{g}}\omega_{j}h(p)=0,\]
where $F_{g}$ denotes the simply connected domain with the boundary $\partial F_{g}=\sum_{j=1}^{g}(\mathcal{A}_{j}+\mathcal{A}_{j}^{-1}+\mathcal{B}_{j}+\mathcal{B}_{j}^{-1})$. By Cauchy's theorem, taking local parameters $k_{a_{j}}$  near $a_{j}$ such that $k_{a_{j}}(p)=f(p)-z_{a}$ for any point $p\in\Rs_{g}$ lying in a neighbourhood of $a_{j}$, we deduce that $\sum_{k=1}^{n+1}\mathbf{V}_{a_{k}}=0$. 
\hspace{5.5cm} $\square$
\\
\\
To prove Theorem 4.1, substitute the functions (\ref{sol n-NLS comp}) into the first equation of (\ref{ass n-NLS}) to get
\[D'_{a_{n+1}}\ln\frac{\Theta(\z+\mathbf{r}_{1})}{\Theta(\z)}+D^{2}_{a_{n+1}}\ln\frac{\Theta(\z+\mathbf{r}_{1})}{\Theta(\z)}
+\left(D_{a_{n+1}}\ln\frac{\Theta(\z+\mathbf{r}_{1})}{\Theta(\z)}-E_{1}\right)^{2}\]
\begin{equation}
+F_{1}-2\sum_{k=1}^{n}q_{2}(a_{n+1},a_{k})\frac{\Theta(\z+\mathbf{r}_{k})\Theta(\z-\mathbf{r}_{k})}{\Theta(\z)^{2}}=0. \label{4.21}
\end{equation}
It can be shown that equation (\ref{4.21}) holds as follows: in (\ref{my corol}), let us choose $a=a_{n+1}$ and $b=a_{1}$ to obtain
\begin{equation}D'_{a_{n+1}}\ln\frac{\Theta(\z+\mathbf{r}_{1})}{\Theta(\z)}+D_{a_{n+1}}^{2}\ln\frac{\Theta(\z+\mathbf{r}_{1})}{\Theta(\z)} +\left(D_{a_{n+1}}\ln\frac{\Theta(\z+\mathbf{r}_{1})}{\Theta(\z)}-K_{1}\right)^{2}+K_{2}+2\,D^{2}_{a_{n+1}}\ln\Theta(\z)=0,  \label{4.22}
\end{equation}
for any $\z\in \C ^{g}$, and in particular for $\z=\mathbf{Z}-\mathbf{d}$;
here we used the notation $K_{i}=K_{i}(a_{n+1},a_{1})$ for $i=1,2$.
By Lemma 4.1 the sum $\sum_{k=1}^{n+1}\mathbf{V}_{a_{k}}$ equals zero, which implies
\[\sum_{k=1}^{n+1}D_{a_{k}}=0.\]
Substituting $D_{a_{n+1}}$ instead of $-\sum_{k=1}^{n}D_{a_{k}}$ in (\ref{4.22}) and using (\ref{cor Fay}) we obtain (\ref{4.21}), where
\[E_{1}=K_{1}, \qquad   F_{1}=K_{2}-2\,\sum_{k=1}^{n}q_{1}(a_{n+1},a_{k}).\]
In the same way, it can be proved that the functions in (\ref{sol n-NLS}) satisfy the $2n-1$ other equations of the system (\ref{ass n-NLS}).
\end{proof}

\vspace{0.3cm}
The solutions (\ref{sol n-NLS comp}) of the complexified sytem (\ref{ass n-NLS}) depend on the Riemann surface $\Rs_{g}$, the meromorphic function $f$ of degree $n+1$, a non critical value $z_{a}\in\C$ of $f$, and arbitrary constants $\mathbf{d}\in\C^{g}$, $A_{j}\neq 0$.
The transformation of the local parameters given by
\begin{equation}
k_{a_{j}}\longrightarrow \beta\,k_{a_{j}}+\mu\,k_{a_{j}}^{2}+O\left(k_{a_{j}}^{3}\right), \label{trans param n-NLS}
\end{equation}
where $\beta,\mu$ are arbitrary complex numbers ($\beta\neq 0$), leads to a different family of solutions of the complexified system (\ref{ass n-NLS}). These new solutions are obtained via the following transformations:
\begin{align}
\psi_{j}(x,t)&\longrightarrow\psi_{j}\left(\beta\,x+2\beta\lambda\,t,\beta^{2}\,t\right)\,\exp\left\{ -\mathrm{i}\left(\lambda\,x+\lambda^{2}\,t\right)\right\},\nonumber\\
\psi^{*}_{j}(x,t)&\longrightarrow\beta^{2}\,\psi^{*}_{j}\left(\beta\,x+2\beta\lambda\,t,\beta^{2}\,t\right)\,\exp\left\{ \mathrm{i}\left(\lambda\,x+\lambda^{2}\,t\right)\right\}, 
\label{trans sol n-NLS}
\end{align}
where $\lambda=\mu\,\beta^{-1}$.

\subsection{Reality conditions}
Algebro-geometric solutions of the n-NLS$^{s}$ equation (\ref{n-NLS bis}) are constructed from solutions (\ref{sol n-NLS comp}) of the complexified system by imposing the reality conditions $\psi_{j}^{*}=s_{j}\,\overline{\psi_{j}}$  (\ref{real cond n-NLS}).

Let $\Rs_{g}$ be a real compact Riemann surface with an anti-holomorphic involution $\tau$. Let us choose the homology basis satisfying (\ref{hom basis}). A meromorphic function $f$ on $\Rs_{g}$ is called real if $f(\tau p)=\overline{f(p)}$ for any $p\in\Rs_{g}$. 

In the next proposition we derive theta-functional solutions of (\ref{n-NLS bis}). The signs $s_{j}$ appearing in the reality conditions (\ref{real cond n-NLS}) are expressed in terms of certain intersection indices on $\Rs_{g}$. These intersection indices are defined as follows: let $f$ be a real meromorphic function
of degree $n+1$ on $\Rs_{g}$. Let $z_{a}\in\R$ be a non critical value of $f$, and assume that the fiber $f^{-1}(z_{a})=\left\{a_{1},\ldots,a_{n+1}\right\}$ over $z_{a}$ belongs to the set $\Rs_{g}(\R)$. Let $\tilde{a}_{n+1},\tilde{a}_{j}\in\Rs_{g}(\R)$ lie in a neighbourhood of $a_{n+1}$ and $a_{j}$ respectively such that $f(\tilde{a}_{n+1})=f(\tilde{a}_{j})$. Denote by $\tilde{\ell}_{j}$ an oriented contour connecting $\tilde{a}_{n+1}$ and $\tilde{a}_{j}$, and having the following decomposition in $H_{1}(\Rs_g\setminus\{a_{n+1},a_{j}\})$ (see Appendix A.2.2)
\begin{equation}
\tau \tilde{\ell}_{j}=\tilde{\ell}_{j}+\mathcal{A}\mathbf{N}_{j}+\mathcal{B}\mathbf{M}_{j}+\alpha_{j}\,\mathcal{S}_{a_{j}},\label{ellj}
\end{equation}
for some $\alpha_{j}\in\Z$, where vectors $\mathbf{N}_{j},\mathbf{M}_{j}\in\Z^{g}$ are the same as in (\ref{hom basis 3}).
Then 
\begin{equation}
\alpha_{j}=(\tau \tilde{\ell}_{j}-\tilde{\ell}_{j})\circ\ell_{j},  \label{alpha}
\end{equation}
between the closed contour $\tau \tilde{\ell}_{j}-\tilde{\ell}_{j}$ and the contour $\ell_{j}$; this intersection is computed in the relative homology group $H_{1}(\Rs_g,\{a_{n+1},a_{j}\})$.
\\\\
Theta functional solutions of (\ref{n-NLS bis}) are given by

\begin{proposition}
Let $f$ be a real meromorphic function
of degree $n+1$ on $\Rs_{g}$. Let $z_{a}\in\R$ be a non critical value of $f$, and assume that the fiber $f^{-1}(z_{a})=\left\{a_{1},\ldots,a_{n+1}\right\}$ over $z_{a}$ belongs to the set $\Rs_{g}(\R)$. Choose the local parameters $k_{a_{j}}$ near $a_{j}$ as $k_{a_{j}}(p)=f(p)-z_{a}$, 
for any point $p\in\Rs_{g}$ lying in a neighbourhood of $a_{j}$. Denote by $\{\mathcal{A},\mathcal{B},\ell_{j}\}$ the standard generators of the relative homology group $H_{1}(\Rs_g,\{a_{n+1},a_{j}\})$ (see Appendix A.2.2). 
Let $\mathbf{d}_{R}\in\R^{g}$, $\mathbf{T}\in\Z^{g}$, and define $\mathbf{d}=\mathbf{d}_{R}+\frac{\mathrm{i}\pi}{2}(\text{diag}(\mathbb{H})-2\,\mathbf{T})$. Morover, take $\theta\in\R$.
Then the following functions $\left\{\psi_{j}\right\}_{j=1}^{n}$ are solutions of n-NLS$^{s}$ (\ref{n-NLS bis}) 
\begin{equation}    
\psi_{j}(x,t)=|A_{j}|\,e^{\mathrm{i}\theta}\,\frac{\Theta(\mathbf{Z}-\mathbf{d}+\mathbf{r}_{j})}{\Theta(\mathbf{Z}-\mathbf{d})}\,\exp \left\{\mathrm{i}(-E_{j}\,x+\,F_{j}\,t)\right\},\label{sol n-NLS}
\end{equation}
where 
$\mathbf{Z}=\mathrm{i}\,\mathbf{V}_{a_{n+1}}\,x+\,\mathrm{i}\,\mathbf{W}_{a_{n+1}}\,t,$ 
and
\begin{equation}
|A_{j}|=|q_{2}(a_{n+1},a_{j})|^{1/2}\,\exp\left\{\tfrac{1}{2}\left\langle\mathbf{d}_{R},\mathbf{M}_{j}\right\rangle\right\}. \label{A n-NLS}
\end{equation} 
Here $\mathbf{r}_{j}=\int_{\ell_{j}}\omega$, the vectors $\mathbf{V}_{a_{n+1}}, \mathbf{W}_{a_{n+1}}$ are defined in (\ref{exp hol diff}), and the vector $\mathbf{M}_{j}\in\Z^{g}$ is defined by the action of $\tau$ on the relative homology group $H_{1}\left(\Rs_{g}^{(n+1)},\{a_{n+1},a_{j}\}\right)$ (see (\ref{hom basis 3})). The scalars $q_2(a_{n+1},a_{j})$ and $E_{j},\,F_{j}$ are introduced in (\ref{q2}) and (\ref{E,N n-NLS}) respectively. The signs $s_{1},\ldots,s_{n}$ are given by 
\begin{equation}
s_{j}= \exp\left\{\mathrm{i}\pi(1+\alpha_{j})+\mathrm{i}\pi\left\langle \mathbf{T} ,\mathbf{M}_{j}\right\rangle\right\}, \label{sj}
\end{equation} 
where the intersection indices $\alpha_{j}\in\Z$ are defined in (\ref{alpha}).
\end{proposition}

\begin{proof}
The proof follows the lines of Section 3.2, where similar statements were proven for the DS1$^{\rho}$ equation. First of all, invariance with respect to the anti-involution $\tau$ of the point $a_{n+1}$ implies the reality of the vector $\mathbf{Z}=\mathrm{i}\,\mathbf{V}_{a_{n+1}}\,x+\,\mathrm{i}\,\mathbf{W}_{a_{n+1}}\,t$.
Moreover, from (\ref{diff hol}) and (\ref{hom basis 3}) we get
\begin{equation}
\overline{\mathbf{r}_{j}}=-\mathbf{r}_{j}-2\mathrm{i}\pi\mathbf{N}_{j}-\mathbb{B}\mathbf{M}_{j}. \label{r n-NLS}
\end{equation}
where $\mathbf{N}_{j},\,\mathbf{M}_{j}\in\Z^{g}$ are defined in (\ref{hom basis 3}) and satisfy 
\begin{equation}
2\,\mathbf{N}_{j}+\mathbb{H}\mathbf{M}_{j}=0. \label{NHM n-NLS}
\end{equation}
For $j=1,\ldots,n$, the action of the complex conjugation on the scalars $K_{1}(a_{n+1},a_{j})$ and $K_{2}(a_{n+1},a_{j})$ is given by (\ref{K1 K2 conj}), and one can directy see from (\ref{cor Fay}) that $q_{1}(a_{n+1},a_{j})$ is real. Hence we get
\begin{equation}
\overline{E_{j}}=E_{j}-\left\langle \mathbf{V}_{a_{n+1}},\mathbf{M}_{j}\right\rangle,  \qquad
\overline{F_{j}}=F_{j}+\left\langle \mathbf{W}_{a_{n+1}},\mathbf{M}_{j}\right\rangle. \label{E N conj}
\end{equation}
Under the assumptions of the theorem and by (\ref{arg q2}), the argument of $q_{2}(a_{n+1},a_{j})$ is given by
\begin{equation}
\arg(q_{2}(a_{n+1},a_{j}))=\pi\left(1+\alpha_{j}+\tfrac{1}{2}\left\langle \mathbb{H}\mathbf{M}_{j},\mathbf{M}_{j}\right\rangle\right)-\frac{1}{2\mathrm{i}}\left(\left\langle \B\,\mathbf{M}_{j},\mathbf{M}_{j}\right\rangle+2\,\left\langle \mathbf{r}_{j} ,\mathbf{M}_{j}\right\rangle\right). \label{arg q2 f}
\end{equation}
Therefore, the reality conditions (\ref{real cond n-NLS}) together with (\ref{sol n-NLS comp}) lead to
\begin{multline}
|A_{j}|^{2}=s_{j}\,|q_{2}(a_{n+1},a_{j})|\,\frac{\Theta(\mathbf{Z}-\mathbf{d}-\mathbf{r}_{j})\,\Theta(\mathbf{Z}-\overline{\mathbf{d}}+\mathrm{i}\pi \,\text{diag}(\mathbb{H}))}{\Theta(\mathbf{Z}-\overline{\mathbf{d}}-\mathbf{r}_{j}+\mathrm{i}\pi \,\text{diag}(\mathbb{H}))\,\Theta(\mathbf{Z}-\mathbf{d})}\\
\times\exp\left\{\mathrm{i}\pi\left(1+\alpha_{j}+\tfrac{1}{2}\left\langle \mathbb{H}\mathbf{M}_{j},\mathbf{M}_{j}\right\rangle\right)+\left\langle\overline{\mathbf{d}}-\mathrm{i}\pi\text{diag}(\mathbb{H}),\mathbf{M}_{j}\right\rangle\right\},\label{4.25}
\end{multline}
if one takes into account (\ref{conj theta}) and (\ref{2.4}).
Let us choose a vector $\mathbf{d}\in \C^{g}$ such that
\[\overline{\mathbf{d}}\equiv\mathbf{d}- \mathrm{i}\pi \,\text{diag}(\mathbb{H}) \,\,\, \left(\text{mod} \,\, (2\mathrm{i}\pi \Z^{g}+ \mathbb{B}\,\Z^{g})\right). \]
Since $\overline{\mathbf{d}}-\mathbf{d}$ is purely imaginary we have
\begin{equation}
 \overline{\mathbf{d}}= \mathbf{d}- \mathrm{i}\pi \,\text{diag}(\mathbb{H}) + 2\mathrm{i}\pi\mathbf{T}, \label{4.26}
\end{equation}
for some $\mathbf{T}\in\Z^{g}$, where we have used (\ref{matrix B}) and the fact that $\B$ has a non-degenerate real part.
It follows that the vector $\mathbf{d}$ can be written as
\begin{equation}
 \mathbf{d}=\mathbf{d}_{R}+\frac{\mathrm{i}\pi}{2}(\text{diag}(\mathbb{H})-2\,\mathbf{T}), \label{vectD n-NLS}
\end{equation}
 for some $\mathbf{d}_{R}\in\R^{g}$ and $\mathbf{T}\in\Z^{g}$.
Therefore, (\ref{4.25}) becomes
\begin{equation}
|A_{j}|^{2}=s_{j}\,|q_{2}(a_{n+1},a_{j})|\exp\left\{\mathrm{i}\pi(1+\alpha_{j})+\tfrac{\mathrm{i}\pi}{2}\left\langle \mathbb{H}\mathbf{M}_{j},\mathbf{M}_{j}\right\rangle+\left\langle\mathbf{d},\mathbf{M}_{j}\right\rangle\right\}, \label{|A| n-NLS}
\end{equation}
which by (\ref{vectD n-NLS}) leads to (\ref{A n-NLS}). Moreover we deduce from (\ref{vectD n-NLS}) and (\ref{|A| n-NLS})  that
\[s_{j}=\exp\left\{\mathrm{i}\pi(1+\alpha_{j})+\tfrac{\mathrm{i}\pi}{2}\left\langle \mathbb{H}\mathbf{M}_{j}+\text{diag}(\mathbb{H}),\mathbf{M}_{j}\right\rangle-\mathrm{i}\pi\left\langle \mathbf{T},\mathbf{M}_{j}\right\rangle\right\}.\]
From (\ref{NHM n-NLS}) and the definition of the matrix $\mathbb{H}$ (see Appendix A.1), it can be deduced that the quantity $\frac{1}{2}\left\langle \mathbb{H}\mathbf{M}_{j}+\text{diag}(\mathbb{H}),\mathbf{M}_{j}\right\rangle$ is even in each case, which yields (\ref{sj}).
\end{proof}

\vspace{0.3cm}
Functions $\psi_{j}$ given in (\ref{sol n-NLS}) describe a family of algebro-geometric solutions of (\ref{n-NLS bis}) depending on: a real Riemann surface $(\Rs_{g},\tau)$, a real meromorphic function $f$ on $\Rs_{g}$ of degree $n+1$, a non critical value $z_{a}\in\R$ of $f$ such that the fiber over $z_{a}$ belongs to the set $\Rs_{g}(\R)$, and arbitrary constants $\mathbf{d}_{R}\in\R^{g}$, $\mathbf{T}\in\Z^{g}$, $\theta\in\R$. Note that the periodicity properties of the theta function imply without loss of generality that the vector $\mathbf{T}$ can be chosen in the set $\{0,1\}^{g}$. The case where the Riemann surface $\Rs_{g}$ is dividing and $\mathbf{T}=0$ is of special importance, because the related solutions are smooth, as explained in Proposition 3.1. In this case, the sign $s_{j}$ (\ref{sj}) is given by $s_{j}= \exp\{\mathrm{i}\pi(1+\alpha_{j})\}$.

\subsection{Solutions of n-NLS$^{+}$ and n-NLS$^{-}$}

Here, we consider the two most physically significant situations: the completely focusing multi-component system n-NLS$^{+}$ (which corresponds to $s=(1,\ldots,1)$), and the completely defocusing system n-NLS$^{-}$ (which corresponds to $s=(-1,\ldots,-1)$).

Starting from a pair $(\Rs_{g},f)$, where $\Rs_{g}$ is a Riemann surface of genus $g$, and where $f$ is a meromorphic function of degree $n+1$ on $\Rs_{g}$, which has $n+1$ simple poles, we construct an $n+1$-sheeted branched covering of $\C\mathbb{P}^{1}$, which we denote by $\Rs_{g,n+1}$. The ramification points of the covering correspond to critical points of $f$; we assume that all of them are simple. 

For any point $a\in\Rs_{g,n+1}$ which is not a critical point  or a pole of the meromorphic function $f$, we use the local parameter $k_{a}(p)=f(p)-f(a)$, for any point $p$  in a neighbourhood of $a$.

According to \cite{EEHS}, by an appropriate choice of the set of generators $\left\{\gamma_{j}\right\}_{j=1}^{2g+2n}$ of the fundamental group $\pi_{1}(\C\mathbb{P}^{1}\setminus \{z_{1},\ldots,z_{2g+2n}\},z_{0})$ of the 
base, which satisfy $\gamma_{1}\ldots\gamma_{2g+2n}=id$, the covering $\Rs_{g,n+1}$ can be represented as follows: consider the hyperelliptic covering of genus $g$ and attach to it $n-1$ spheres as shown in Figure 1. More precisely, the generators $\gamma_{j}$ can be chosen in such way that the loop $\gamma_{j}$ encircles only the point $z_{j}$; the corresponding elements $\sigma_{j}\in\mathbf{S}_{n+1}$ (where $\mathbf{S}_{n+1}$ denotes the symmetric group of order $n+1$) of the monodromy group of the covering are given by
\begin{align}
\sigma_{j}=&(n+1,n), & \,\,\,\,j=1,\ldots,2g+2,\nonumber\\
\sigma_{2g+2+2k+1}=\sigma_{2g+2+2k+2}=&(n-k,n-k-1), & k=0,\ldots,n-2. \nonumber
\end{align}

We denote by $x_{1},\ldots,x_{2g+2n}\in\Rs_{g,n+1}$ the critical points of the meromorphic function $f$, and by $z_{j}=f(x_{j})\in\C$ the critical values.
\begin{figure}[hbtp]  
    \label{Fig1}
\begin{center}
\includegraphics[height=45mm]{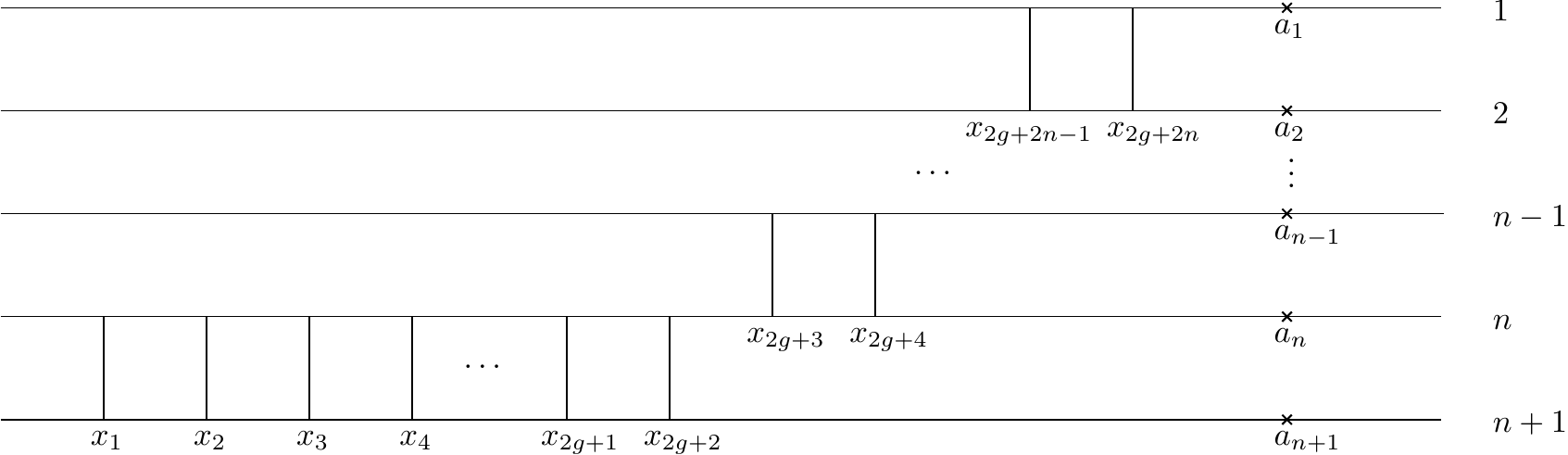}
\end{center}
\caption{\textit{Hurwitz diagram of the covering $\Rs_{g,n+1}$.}}
\end{figure}
Assume that the branch points $\left\{z_{j}\right\}_{j=1}^{2g+2n}$ are real or pairwise conjugate, and order them as follows:
\[\text{Re}(z_{1})\leq\ldots\leq \text{Re}(z_{2g+2n}).\]
Let us introduce an anti-holomorphic involution $\tau$ on $\Rs_{g,n+1}$, which acts as the complex conjugation on each sheet.

\vspace{0.5cm}
\subsubsection{Solutions of n-NLS$^{+}$.} Here we construct solutions of the n-NLS$^{+}$ system
\begin{equation}		
\mathrm{i}\,\frac{\partial \psi_{j}}{\partial t}+\frac{\partial^{2} \psi_{j}}{\partial x^{2}}+2\,\left(\sum_{k=1}^{n}|\psi_{k}|^{2}\right)\,\psi_{j} =0,  \quad \quad j=1,\ldots,n.   \label{n-NLS+}
\end{equation}
Let us first describe the covering and the homology basis used in the construction of the solutions.

Assume that all branch points of the covering $\Rs_{g,n+1}$ are pairwise conjugate. Denote this covering by $\Rs_{g,n+1}^{+}$, refering to the focusing system (\ref{n-NLS+}). The covering $\Rs_{g,n+1}^{+}$ admits two real ovals if the genus $g$ is odd, and only one if $g$ is even. Each of them consists of a closed contour on the covering having a real projection into the base.
It is straightforward to see that the covering $\Rs_{g,n+1}^{+}$ is dividing (see Appendix A.1): two points which have respectively a positive and a negative imaginary projection onto $\C$, cannot be connected by a contour which does not cross a real oval. Hence the set of fixed points of the anti-holomorphic involution $\tau$ separates the covering into two connected components.

Now let us choose the canonical homology basis such that all basic cycles belong to sheets $n+1$ and $n$, and such that the anti-holomorphic involution $\tau$ acts on them as in (\ref{hom basis}). By the previous topological description of $\Rs_{g,n+1}^{+}$, the matrix $\mathbb{H}$ involved in (\ref{hom basis}) looks as:
\[\mathbb{H}={\left(\begin{array}{cccccccc}
0&1&&&&\\
1&0&&&&\\
&&\ddots&&&\\
&&&0&1&\\
&&&1&0&\\
&&&&&0\\
\end{array}\right)} \quad \text{if $g$ is odd},\]\[ \mathbb{H}={\left(\begin{array}{cccccccc}
0&1&&&\\
1&0&&&\\
&&\ddots&&\\
&&&0&1\\
&&&1&0\\
\end{array}\right)} \quad \text{if $g$ is even}.\]
The canonical homology basis is described explicitely in Figure 2 for odd genus, and in Figure 3 for even genus.

\begin{figure}[hbtp]  
    \label{Fig2}
\begin{center}
\includegraphics[height=90mm]{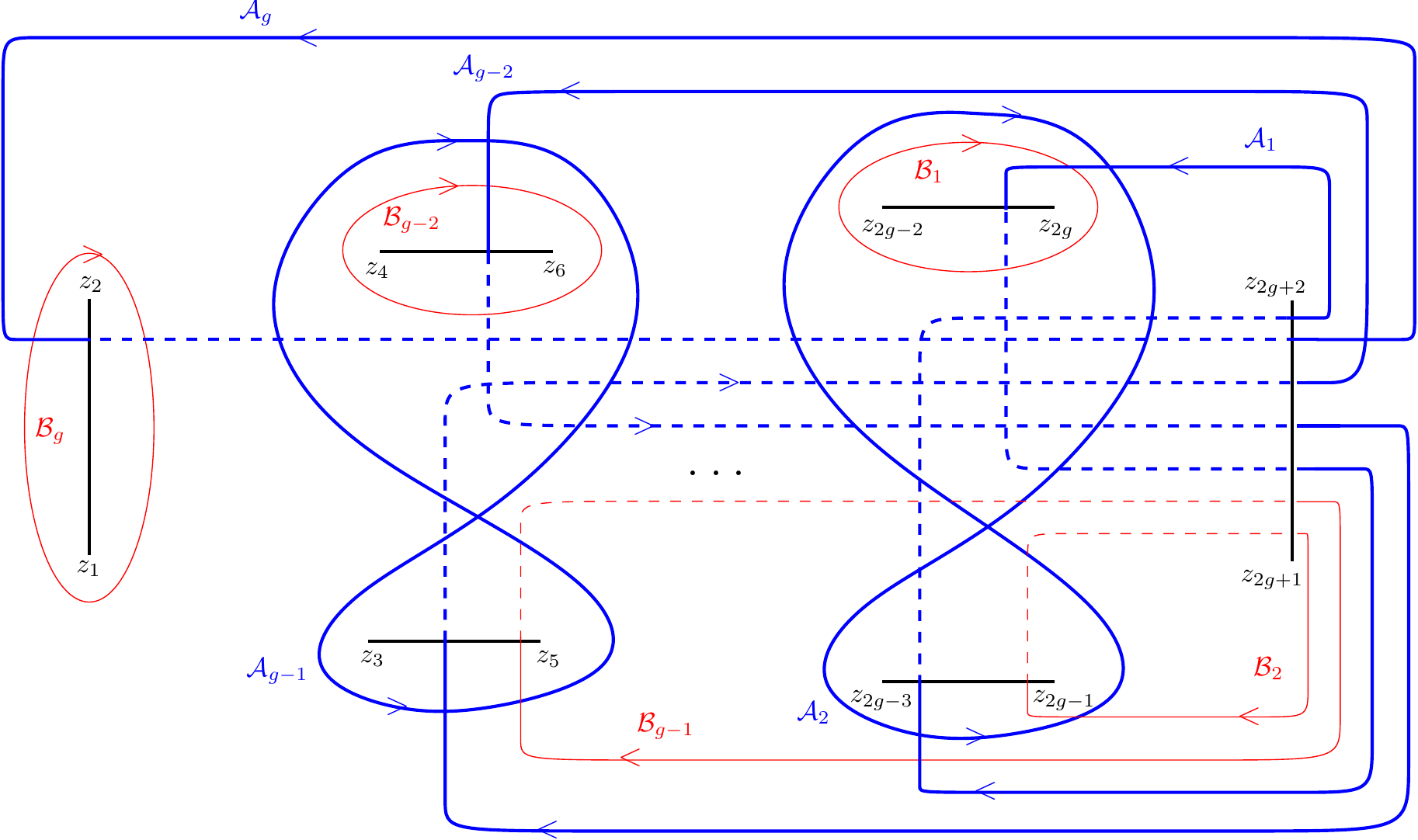}
\end{center}
\caption{\textit{Homology basis on the covering $\Rs_{g,n+1}^{+}$ when the genus $g$ is odd. The solid line
indicates the sheet $n+1$, and the dashed line sheet $n$.}}
\end{figure}

\begin{figure}[hbtp]  
    \label{Fig3}
\begin{center}
\includegraphics[height=90mm]{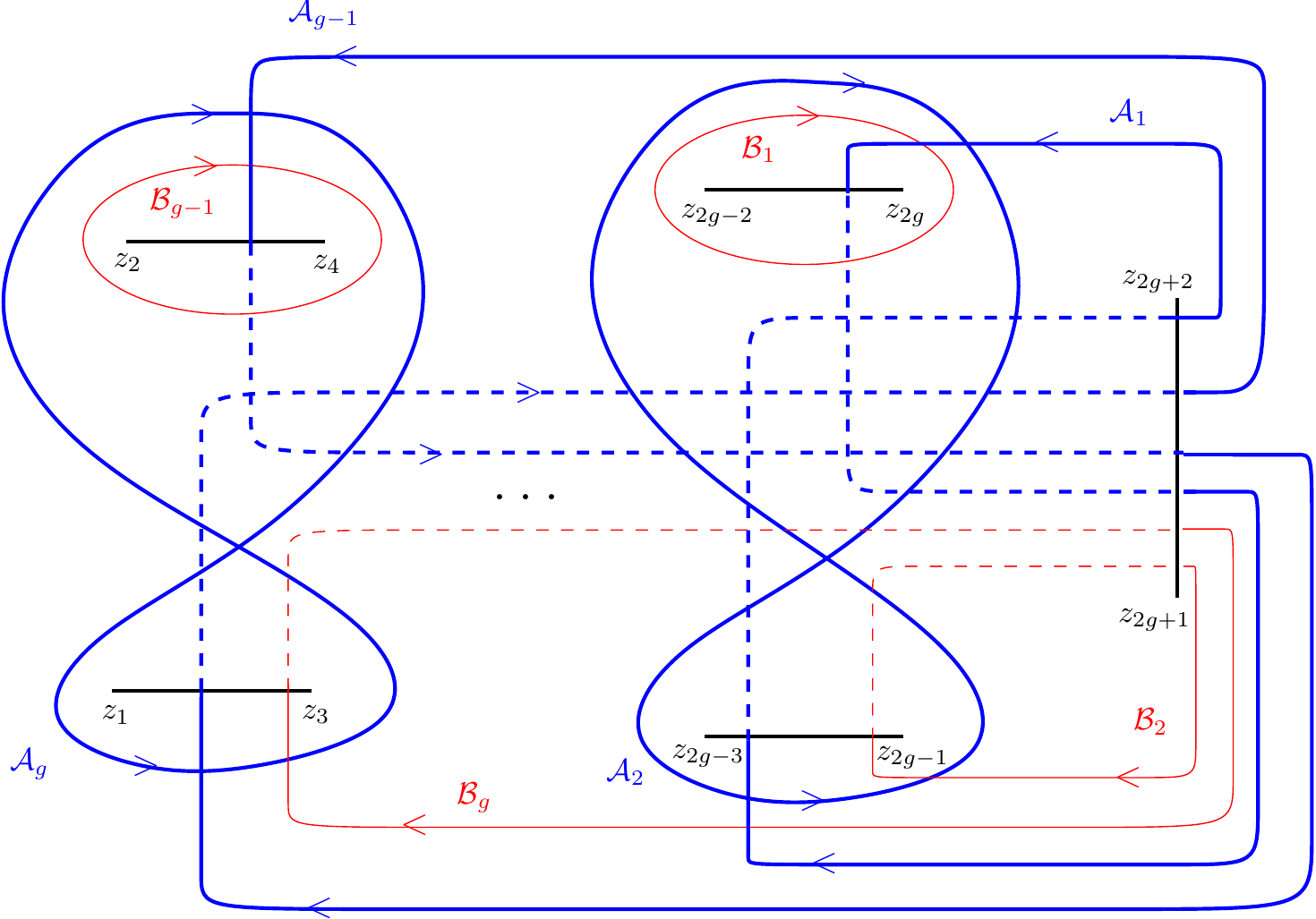}
\end{center}
\caption{\textit{Homology basis on the covering $\Rs_{g,n+1}^{+}$ when the genus $g$ is even. The solid line
indicates the sheet $n+1$, and the dashed line sheet $n$.}}
\end{figure}

As proved in the following theorem, among all coverings having a monodromy group described in Figure 1, only the covering $\Rs_{g,n+1}^{+}$ leads to algebro-geometric solutions of the focusing system (\ref{n-NLS+}).

\begin{theorem}
Consider the covering $\Rs_{g,n+1}^{+}$ and the canonical homology basis discussed above. Fix $z_{a}\in\R$ such that $z_{a}>\text{Re}(z_{j})$ for $j=1,\ldots,2g+2n$. Consider the fiber $f^{-1}(z_{a})=\left\{a_{1},\ldots,a_{n+1}\right\}$ over $z_{a}$, where $a_{j}\in \Rs_{g,n+1}^{+}(\R)$ belongs to sheet $j$ (each of the $a_{j}$ is invariant under the involution $\tau$). Let $\mathbf{d}\in\R^{g}$ and $\theta\in\R$. Then the following functions $\left\{\psi_{j}\right\}_{j=1}^{n}$ are smooth solutions of n-NLS$^{+}$:
\begin{equation}    
\psi_{j}(x,t)=|A_{j}|\,e^{\mathrm{i}\theta}\,\frac{\Theta(\mathbf{Z}-\mathbf{d}+\mathbf{r}_{j})}{\Theta(\mathbf{Z}-\mathbf{d})}\,\exp \left\{\mathrm{i}(-E_{j}\,x+\,F_{j}\,t)\right\},\label{sol n-NLS+}
\end{equation}
where $\mathbf{Z}=\mathrm{i}\,\mathbf{V}_{a_{n+1}}\,x+\,\mathrm{i}\,\mathbf{W}_{a_{n+1}}\,t.$
Here $\mathbf{r}_{j}=\int_{a_{n+1}}^{a_{j}}\omega$, the vectors $\mathbf{V}_{a_{n+1}},\,\mathbf{W}_{a_{n+1}}$ are defined in (\ref{exp hol diff}), and the vector $\mathbf{M}_{j}\in\Z^{g}$ is defined in (\ref{hom basis 3}), according to the action of $\tau$ on the relative homology group $H_{1}\left(\Rs_{g,n+1}^{+},\{a_{n+1},a_{j}\}\right)$. The scalars $|A_{j}|$ and $E_{j},\,F_{j}$ are given by (\ref{A n-NLS}) and (\ref{E,N n-NLS}) respectively.
\end{theorem}

\begin{proof} 
Let us check that the conditions of the theorem imply that functions $\psi_{j}$ in (\ref{sol n-NLS}) are solutions of n-NLS$^{s}$ for $s=(1,\ldots,1)$. Since the matrix $\mathbb{H}$ associated to the covering $\Rs_{g,n+1}^{+}$ satisfies $\text{diag}(\mathbb{H})=0$, and $\mathbf{d}\in\R^{g}$ (i.e. $\mathbf{T}=0$), the quantities $\left\{s_{j}\right\}_{j=1}^{n}$ (\ref{sj}) become
\begin{equation}
s_{j}=\exp\left\{\mathrm{i}\pi(1+\alpha_{j})\right\}. \label{sj foc}
\end{equation} 
Let us first compute the intersection index $\alpha_{n}$. Let $\tilde{a}_{n+1},\,\tilde{a}_{n}\in\Rs_{g,n+1}^{+}(\R)$ lie in a neighbourhood of $a_{n+1}$ and $a_{n}$ respectively such that $f(\tilde{a}_{n+1})=f(\tilde{a}_{n})=z_{\tilde{a}}$. Denote by $\tilde{\ell}_{n}$ an oriented contour connecting $\tilde{a}_{n+1}$ and $\tilde{a}_{n}$. Then the intersection index $\alpha_{n}$ between the closed contour $\tau \tilde{\ell}_{n}-\tilde{\ell}_{n}$ and the contour $\ell_{n}$ satisfies (see Figure 4)
\begin{figure}[hbtp]  
    \label{Fig4}
\begin{center}
\includegraphics[height=75mm]{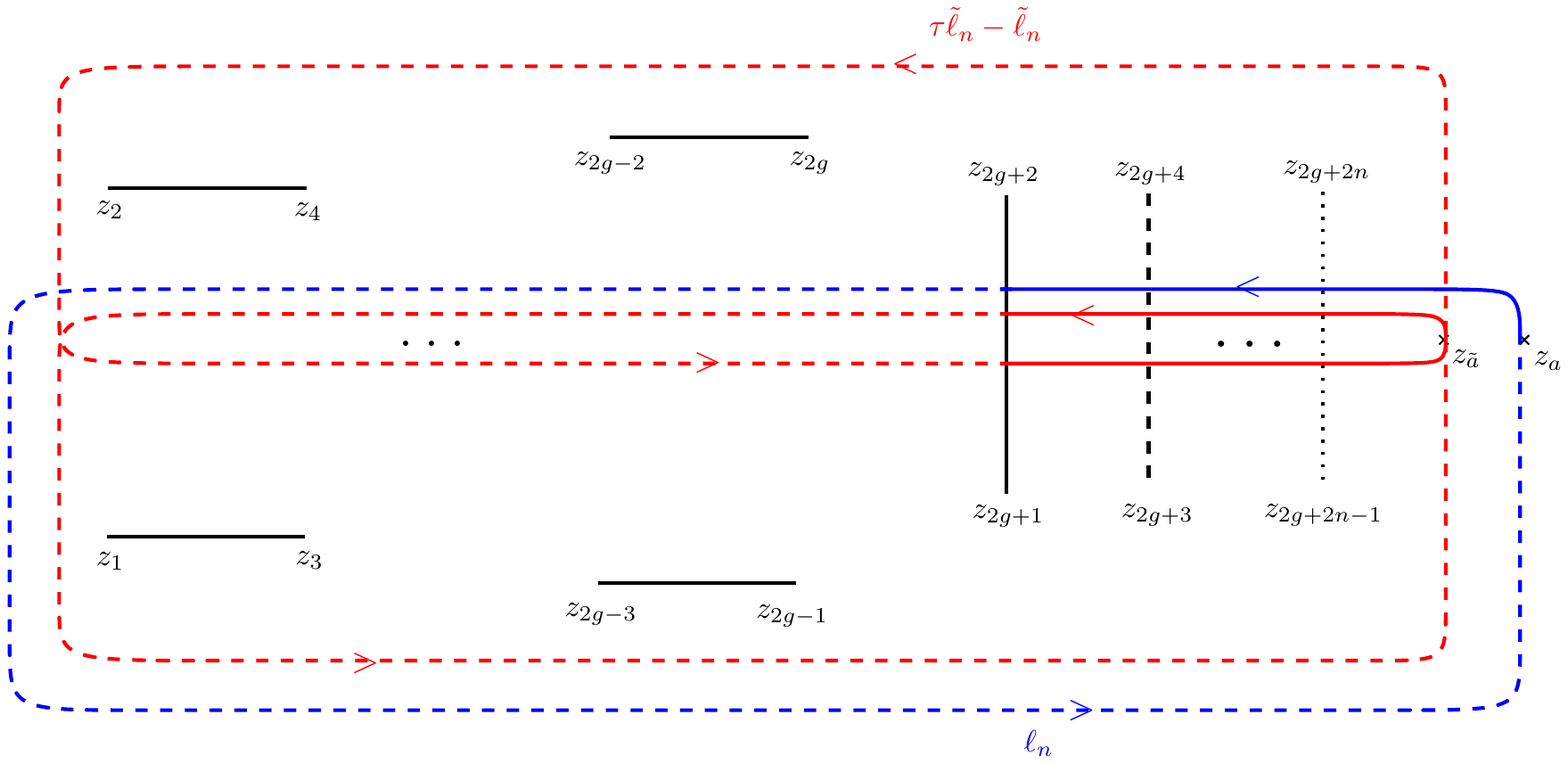}
\end{center}
\caption{\textit{The closed contour $\tau \tilde{\ell}_{n}-\tilde{\ell}_{n}\in H_{1}\left(\Rs_{g,n+1}^{+}\setminus\{a_{n+1},a_{n}\}\right)$ is homologous to a closed contour which encircles the vertical cut $[z_{2g+1},z_{2g+2}]$, then $\alpha_{n}=(\tau \tilde{\ell}_{n}-\tilde{\ell}_{n})\circ \ell_{n}=1$.}}
\end{figure}
\\
\begin{equation}
\alpha_{n}=(\tau \tilde{\ell}_{n}-\tilde{\ell}_{n})\circ \ell_{n}\equiv 1 \quad (\text{mod}\,\, 2), \label{index foc}
\end{equation}
which leads to $s_{n}=1$. Intersection indices $\alpha_{j}$ for $j=1,\ldots,n-1$ can be computed in the same way. 
Therefore
\[\alpha_{1}\equiv\alpha_{2}\equiv\ldots\equiv\alpha_{n}\equiv 1 \quad (\text{mod}\,\, 2),\]
which implies $s_{j}=1$. By Proposition 3.1, smoothness of the solutions is ensured by the reality of the vector $\mathbf{Z}-\mathbf{d}$ and the fact that the curve is dividing.
\end{proof}

\vspace{0.3cm}
Functions $\psi_{j}$ given in (\ref{sol n-NLS+}) describe a family of smooth algebro-geometric solutions of the focusing multi-component NLS equation depending  on $g+n$ complex parameters: $z_{2k-1}\in\C\setminus\R$ for $k=1,\ldots,g+n$; and $g+2$ real parameters: $z_{a},\theta\in\R$, and  $\mathbf{d}\in\R^{g}$.

\subsubsection{Solutions of n-NLS$^{-}$.} Now let us construct solutions of the system n-NLS$^{-}$
\begin{equation}		
\mathrm{i}\,\frac{\partial \psi_{j}}{\partial t}+\frac{\partial^{2} \psi_{j}}{\partial x^{2}}-2\,\left(\sum_{k=1}^{n}|\psi_{k}|^{2}\right)\,\psi_{j} =0,  \quad \quad j=1,\ldots,n.   \label{n-NLS-}
\end{equation}
As for the focusing case, let us first describe the covering and the homology basis used in our construction of the solutions of (\ref{n-NLS-}).

Assume that the branch points $z_{k}$ of the covering $\Rs_{g,n+1}$ are real for $k=1,\ldots,g+2$, and that the branch points $z_{k},\,z_{k+1}$ are pairwise conjugate for $k=2g+3,\ldots,2g+2n$.
Denote by $\Rs_{g,n+1}^{-}$ this covering, refering to the defocusing system (\ref{n-NLS-}). It is straightforward to see that such a covering is an M-curve (see Appendix A.1), that is it admits a maximal number of real ovals $g+1$ with respect to the anti-holomorphic involution $\tau$. On the other hand, it can be directly seen that $\Rs_{g,n+1}^{-}$ is dividing: two points which lie on the sheet $n+1$ and have respectively a positive and a negative imaginary projection onto $\C$ cannot be connected by a contour which does not cross a real oval.

Now let us choose the canonical homology basis such that all basic cycles belong to sheets $n+1$ and $n$, and which satisfies (\ref{hom basis}). Since the covering $\Rs_{g,n+1}^{-}$ is an M-curve, the matrix $\mathbb{H}$ involved in (\ref{hom basis}) satisfies $\mathbb{H}=0$. Such a canonical homology basis is shown in Figure 5.

\begin{figure}[hbtp]  
    \label{Fig5}
\begin{center}
\includegraphics[height=30mm]{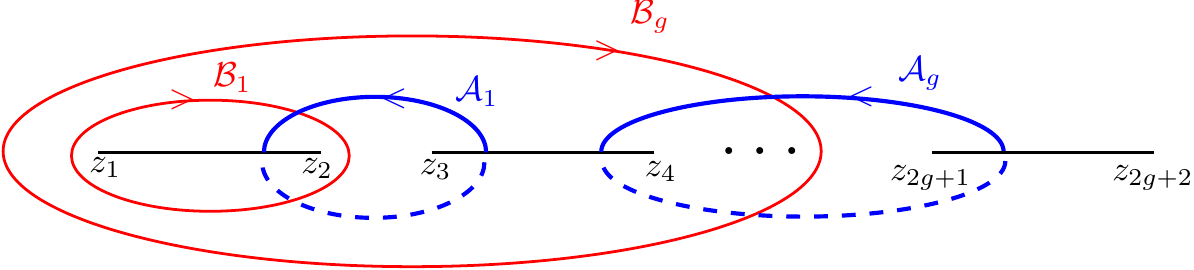}
\end{center}
\caption{\textit{Homology basis on the covering $\Rs_{g,n+1}^{-}$. The solid line
indicates the sheet $n+1$, and the dashed line sheet $n$.}}
\end{figure}
\vspace{0.3cm}

In the following theorem, we construct algebro-geometric solutions of the defocusing system (\ref{n-NLS-}) associated to the covering $\Rs_{g,n+1}^{-}$.

\begin{theorem}
Consider the covering $\Rs_{g,n+1}^{-}$ and the canonical homology basis discussed above.
Fix $z_{a}\in\R\setminus\left\{z_{1},\ldots,z_{2g+2}\right\}$ such that $z_{a}>\text{Re}(z_{j})$ for $j=1,\ldots,2g+2n$. Consider the fiber $f^{-1}(z_{a})=\left\{a_{1},\ldots,a_{n+1}\right\}$ over $z_{a}$, where $a_{j}\in \Rs_{g,n+1}^{-}(\R)$ belongs to sheet $j$ (each of the $a_{j}$ is invariant under the involution $\tau$). Let $\mathbf{d}\in\R^{g}$ and $\theta\in\R$. Then the functions $\left\{\psi_{j}\right\}_{j=1}^{n}$ in (\ref{sol n-NLS+}) are smooth solutions of n-NLS$^{-}$.
\end{theorem}

\begin{proof}
Analogously to the focusing case, one has to check that all $s_{j}=-1$. Since all branch points $z_{k}$ are real for $k=1,\ldots,2g+2$, the intersection index $\alpha_{n}$ between the closed contour $\tau \tilde{\ell}_{n}-\tilde{\ell}_{n}$ and the contour $\ell_{n}$ satisfies (see Figure 6)
\begin{equation}
\alpha_{n}=(\tau \tilde{\ell}_{n}-\tilde{\ell}_{n})\circ \ell_{n}\equiv 0 \quad (\text{mod}\,\, 2), \label{index defoc}
\end{equation}
which leads to $s_{n}=-1$. Intersection indices $\alpha_{j}$ for $j=1,\ldots,n-1$ can be computed in the same way, and we get
\[\alpha_{1}\equiv\alpha_{2}\equiv\ldots\equiv\alpha_{n}\equiv 0 \quad (\text{mod}\,\, 2),\]
which implies $s_{j}=-1$. Smoothness of the solutions is ensured by the reality of the vector $\mathbf{Z}-\mathbf{d}$ and the fact that the curve is dividing.
\end{proof}
\begin{figure}[hbtp]  
    \label{Fig6}
\begin{center}
\includegraphics[height=60mm]{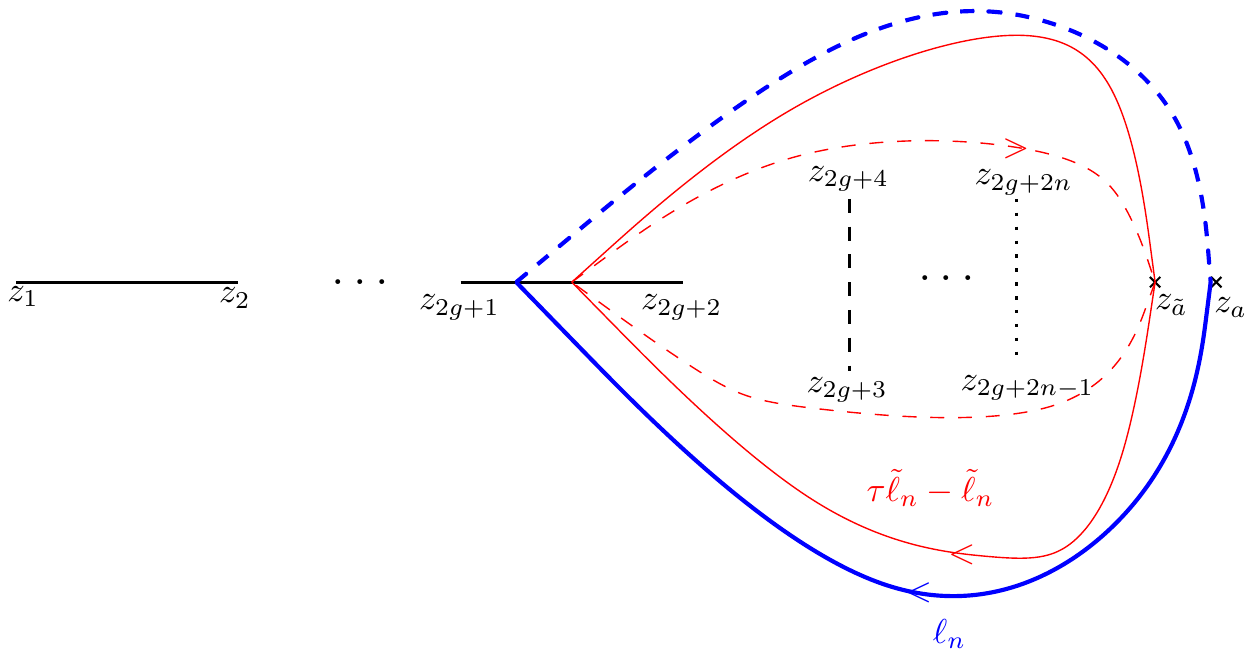}
\end{center}
\caption{\textit{The closed contour $\tau \tilde{\ell}_{n}-\tilde{\ell}_{n}\in H_{1}\left(\Rs_{g,n+1}^{-}\setminus\{a_{n+1},a_{n}\}\right)$ is homologous to zero, then $\alpha_{n}=(\tau \tilde{\ell}_{n}-\tilde{\ell}_{n})\circ \ell_{n}=0$.}}
\end{figure}
\vspace{0.3cm}

Solutions $\psi_{j}$ construced here describe a family of smooth algebro-geometric solutions of the defocusing multi-component NLS equation depending on $n-1$ complex parameters: $z_{2g+2+2k-1}\in\R$ for $k=1,\ldots,n-1$; and $3g+4$ real parameters: $z_{k}\in\R$ for $k=1,\ldots,2g+2$, $z_{a},\,\theta\in\R$, and  $\mathbf{d}\in\R^{g}$.

\begin{remark}
\rm{Smooth solutions of n-NLS$^{s}$ for a vector $s$ with mixed signs can be constructed in the same way.}
\end{remark}

\subsection{Stationary solutions of n-NLS}

It is well-known that the algebro-geometric solutions (\ref{sol n-NLS}) on an elliptic surface describe \textit{travelling waves}, i.e., the modulus of the corresponding solutions depends only on
$x-ct$, where $c$ is a constant. Due to the \textit{Galilei invariance} of the   multi-component NLS equation (see (\ref{trans sol n-NLS})), the invariance under 
transformations of the form 
\[\psi_{j}(x,t)\longrightarrow \psi_{j}(x+2\lambda\,t,t)\,\exp\left\{ -\mathrm{i}\left(\lambda\,x+\lambda^{2}\,t\right)\right\},\]
where $\lambda=-\frac{1}{2}\,W_{a_{n+1}}\,(V_{a_{n+1}})^{-1}$, leads to stationary solutions ($t$-independent) in the transformed coordinates.

For arbitrary genus of the spectral curve, stationary solutions of the multi-component NLS equation are obtained from solutions (\ref{sol n-NLS}) under the vanishing condition 
\begin{equation}
\mathbf{W}_{a_{n+1}}=0. \label{stat sol}
\end{equation}
This condition is equivalent to the existence of a meromorphic function $h$ of order two on $\Rs_{g}$, such that the point $a_{n+1}$ is a critical point of $h$ (this can be proved analogously to Lemma 4.1). 

Therefore, stationary solutions of the multi-component NLS can be constructed from the algebro-geometric data 
$(\Rs_{g},f,h,z_{a})$, where:\\
$\bullet$ $\Rs_{g}$ is a real Riemann surface of genus $g$, and $f$ is a real meromorphic function of order $n+1$ on $\Rs_{g}$,\\
$\bullet$ $z_{a}\in\C\mathbb{P}^{1}$ is a non critical value of $f$ such that $f^{-1}(z_{a})=\left\{a_{1},\ldots,a_{n+1}\right\}$,\\
$\bullet$ $h$ is a real meromorphic function of order two on $\Rs_{g}$, and $a_{n+1}$ is a critical point of $h$,\\
$\bullet$ for $j=1,\ldots,n$, local parameters $k_{a_{j}}$ near $a_{j}$ are chosen to be $k_{a_{j}}(p)=h(p)-h(a_{j})$ for any point $p$ lying in a neighbourhood of $a_{j}$, and $k_{a_{n+1}}(p)=(h(p)-h(a_{n+1}))^{1/2}$ for any point $p$ lying in a neighbourhood of $a_{n+1}$.

With this choice of local parameters, we get $f(p)-z_{a}=\beta_{j}\,k_{a_{j}}(p)+\mu_{j}\,\,k_{a_{j}}(p)^2+O\left(k_{a_{j}}(p)^3\right)$, 
for any point $p\in\Rs_{g}$ which lies in a neighbourhood of $a_{j}$, where $\beta_{j},\mu_{j}\in\R$. Hence solutions (\ref{sol n-NLS}) can be rewritten using this choice of local parameters and then are expressed by the use of the scalars $\beta_{j}$ and $\mu_{j}$. 

Moreover, choosing $a_{n+1}$ as a critical point of $h$, we get (\ref{stat sol}). In this case, the modulus of solutions (\ref{sol n-NLS}) do not depend on the variable $t$.

\subsection{Reduction of n-NLS to (n-1)-NLS}

It is natural to ask if starting from solutions of n-NLS we can obtain solutions of (n-1)-NLS for $n>2$. Such a reduction is possible if  one of the functions $\psi_{j}$ solutions of n-NLS vanishes identically.

Let $\Rs_{g,n+1}^{+}$ be the $(n+1)$-sheeted covering introduced in Section 4.3.1; to obtain solutions of (n-1)-NLS$^{+}$ from solutions of n-NLS$^{+}$, we consider the following degeneration of the covering $\Rs_{g,n+1}^{+}$: let the branch points $z_{2g+2n}$ and $z_{2g+2n-1}$ coalesce, in such way that the first sheet gets disconnected from the other sheets (see Figure 1); denote by $\Rs_{g,n}^{+}$ the covering obtained in this limit.

Then the normalized holomorphic differentials on $\Rs_{g,n+1}^{+}$ tend to normalized holomorphic differentials on $\Rs_{g,n}^{+}$; on the first sheet, all holomorphic differentials tend to zero. Therefore, in this limit, each component of the vector $\mathbf{V}_{a_{1}}$ tends to $0$.

Hence by (\ref{q2}) and (\ref{A n-NLS}), the function $\psi_{1}$ tends to zero as $z_{2g+2n}$ and $z_{2g+2n-1}$ coalesce. Functions $\left\{\psi_{j}\right\}_{j=2}^{n}$ obtained in this limit are solutions of (n-1)-NLS$^{+}$ associated to the covering $\Rs_{g,n}^{+}$.

A similar degeneration produces a solution of (n-1)-NLS$^{-}$ from a solution of n-NLS$^{-}$.

\begin{remark} \rm{Repeating this degeneration $n-3$ times, we rediscover (see \cite{Its}) algebro-geometric solutions of the focusing (resp. defocusing) non-linear Schrödinger equation (\ref{NLS}) associated to an hyperelliptic curve with pairwise conjugate branch points (resp. real branch points).}
\end{remark}

\subsection{Relationship between solutions of KP1 and solutions of n-NLS}
Historically, the Korteweg-de Vries equation (KdV) and its generalization to two spatial variables, the Kadomtsev-Petviashvili equations (KP), were the most important examples of applications of methods of algebraic geometry in the 1970's (see e.g. \cite{BBEIM}). Moreover, the KP equation is the first example of a system with two space variables for which it has been possible to completely solve the problem of reality of algebro-geometric solutions.

Here we show that starting from our solutions of the multi-component NLS equation and its complexification, we can construct a subclass of complex and real solutions of the Kadomtsev-Petviashvili equation (KP1) 
\begin{equation}
\frac{3}{4}\,u_{yy}=\left(u_{t}-\frac{1}{4}\,(6\,u\,u_{x}-u_{xxx})\right)_{x}. \label{c}
\end{equation}

Let $\Rs_{g}$ be an arbitrary Riemann surface with marked point $a$, and let $k_{a}$ be an arbitrary local parameter near $a$. Define vectors $\mathbf{V}_{a},\,\mathbf{W}_{a},\,\mathbf{U}_{a}$ as in (\ref{exp hol diff}) and let $\mathbf{d}\in\C^{g}$. Then, according to Krichever's theorem \cite{Kr}, the function 
\begin{equation}
u(x,y,t)=2\,D_{a}^{2}\log \Theta(\mathrm{i}\,\mathbf{V}_{a}\,x+\mathrm{i}\,\mathbf{W}_{a}\,y+\mathrm{i}\,\mathbf{U}_{a}\,t+\mathbf{d})+2\,c  \label{u}
\end{equation}
is a solution of KP1; here the constant $c$ is defined by the expansion near $a$ of the normalized meromorphic differential $\Omega_{a}^{(2)}(p)$ having a pole of order two at $a$ only: $\Omega_{a}^{(2)}(p)=(k_{a}(p))^{-2}+c\,k_{a}(p)+...$, where $p$ lies in a neighbourhood of $a$. 

Let us check that if the local parameter $k_{a}$ is defined by the meromorphic function $f$ as $k_{a}(p)=f(p)-f(a)$, then formula (\ref{u}) naturally arises from our construction of solutions of the n-NLS$^{s}$ system. Namely, identify $a$ with $a_{n+1}$. Then, due to the fact that $\sum_{j=1}^{n+1} \mathbf{V}_{a_{j}}=0$ (see Lemma 4.1), the solution (\ref{u}) of KP1 can be rewritten as
\[u(x,y,t)=-2\,\sum_{j=1}^{n}D_{a_{n+1}}D_{a_{j}}\log \Theta(\z)+2\,c,\]
where $\z=\mathrm{i}\,\mathbf{V}_{a_{n+1}}\,x+\mathrm{i}\,\mathbf{W}_{a_{n+1}}\,y+\mathrm{i}\,\mathbf{U}_{a_{n+1}}\,t+\mathbf{d}$. Using corollary (\ref{cor Fay}) of Fay's identity, we get
\begin{equation}
u(x,y,t)=-2\,\sum_{j=1}^{n}\left(q_{1}(a_{n+1},a_{j})+q_{2}(a_{n+1},a_{j})\,\frac{\Theta(\z+\mathbf{r}_{j})\,\Theta(\z-\mathbf{r}_{j})}{\Theta(\z)^{2}}\right)+2\,c. \label{4.37}
\end{equation}

Now let us consider solutions $\psi_{j},\,\psi_{j}^{*}$ (\ref{sol n-NLS comp}) of the complexified multi-component NLS equation,
and make the change of variables $(x,t)\rightarrow (x,y)$ and $\mathbf{d}\rightarrow -\mathrm{i}\,\mathbf{U}_{a_{n+1}}\,t+\mathbf{d}$. Then by (\ref{4.37}), the complex-valued solutions $u$ (\ref{u}) of KP1 and solutions $\psi_{j},\,\psi_{j}^{*}$ (\ref{sol n-NLS comp}) of the complexified n-NLS system are related by
\begin{equation}
u(x,y,t)=\gamma-2\,\sum_{j=1}^{n}\psi_{j}(x,y,t)\,\psi_{j}^{*}(x,y,t), \label{u psi}
\end{equation}
where
\[\gamma=-2\,\sum_{j=1}^{n}q_{1}(a_{n+1},a_{j})+2\,c.\]

If we impose the reality conditions (\ref{real cond n-NLS}), we obtain real solutions (\ref{u}) of KP1 from our solutions (\ref{sol n-NLS+}) of n-NLS$^{s}$ equation
\begin{equation}
u(x,y,t)=\gamma-2\,\sum_{j=1}^{n}s_{j}\,|\psi_{j}(x,y,t)|^{2}. \label{u psi bis}
\end{equation}

Due to the fact that in our construction of solutions of the multi-component NLS equation, the local parameters are defined by the meromorphic function $f$, complex solutions (\ref{u psi}) and real solutions (\ref{u psi bis}) of KP1 obtained in this way form only a subclass of Krichever's solutions.

\vspace{2cm}

\thanks{I thank C.~Klein, who interested me in the subject, and D.~Korotkin for carefully reading the manuscript and providing valuable hints. I am grateful to B.~Dubrovin and V.~Shramchenko for useful discussions.
This work has been supported in part by the project FroM-PDE funded by the European
Research Council through the Advanced Investigator Grant Scheme, the Conseil R\'egional de Bourgogne
via a FABER grant and the ANR via the program ANR-09-BLAN-0117-01. }

\appendix

\section{Real Riemann surfaces}

In this section, we recall some facts from the theory of real compact Riemann 
surfaces. Following \cite{Vin}, we introduce a symplectic basis of cycles on $\mathcal{R}_{g}$ and study reality properties of various objects   
 on the Riemann surface  $\mathcal{R}_{g}$ associated to this basis.

\subsection{Action of $\tau$ on the homology group $H_{1}(\Rs_{g})$}

A Riemann surface $\mathcal{R}_{g}$ is called real  if it admits an 
anti-holomorphic involution 
$\tau:\mathcal{R}_{g}\rightarrow\mathcal{R}_{g}, \,\tau^{2}=1$. The 
connected components of the set of
fixed points of the anti-involution $\tau$ are called real ovals of $\tau$. We denote 
by $\mathcal{R}_{g}(\R)$ the set of fixed points. Assume that
$\mathcal{R}_{g}(\R)$ consists of $k$ real ovals, 
with $0\leq k\leq g+1$. The curves with the maximal number of real ovals, $k=g+1$, are called M-curves.

The complement $\mathcal{R}_{g}\setminus 
\mathcal{R}_{g}(\R)$ has either one or two connected components. The curve $\mathcal{R}_{g}$ is called a \textit{dividing} curve (or that 
$\mathcal{R}_{g}$ divides) if $\mathcal{R}_{g}\setminus\mathcal{R}_{g}(\R)$ has two components, and $\mathcal{R}_{g}$ is called \textit{non-dividing} if $\mathcal{R}_{g}\setminus \mathcal{R}_{g}(\R)$ is connected (notice that an M-curve is always a dividing curve). 

\begin{example}
Consider the hyperelliptic Riemann surface of genus $g$ defined by the equation
\begin{equation}
\mu^{2}= \prod_{k=1}^{2g+1} (\lambda-\lambda_{k}), \label{3.2}
\end{equation}
where the branch points $\lambda_{k}\in\R$ are ordered such that $\lambda_{1}<\ldots<\lambda_{2g+1}$. On such a Riemann surface, we can define two anti-holomorphic involutions $\tau_{1}$ and $\tau_{2}$, given respectively by $\tau_{1}(\lambda,\mu)=(\overline{\lambda},\overline{\mu})$ and $\tau_{2}(\lambda,\mu)=(\overline{\lambda},-\overline{\mu})$. Projections of real ovals of $\tau_{1}$  on the $\lambda$-plane coincide with  the intervals $[\lambda_{1},\lambda_{2}],\ldots,[\lambda_{2g+1},+\infty]$, and projections of real ovals of $\tau_{2}$  on the $\lambda$-plane coincide with  the intervals 
$[-\infty,\lambda_{1}],\ldots,[\lambda_{2g},\lambda_{2g+1}]$. Hence the curve (\ref{3.2}) is an M-curve with respect to both anti-involutions $\tau_{1}$ and $\tau_{2}$. 
\end{example}

\vspace{0.3cm}
Denote by $\left\{\mathbf{\mathcal{A}},\mathbf{\mathcal{B}}\right\}$ the set of generators of the homology group $H_{1}(\Rs_{g})$, where $\mathbf{\mathcal{A}}=(\mathcal{A}_{1},\ldots, \mathcal{A}_{g})^{T}$ and $\mathbf{\mathcal{B}}=(\mathcal{B}_{1},\ldots, \mathcal{B}_{g})^{T}$. 
According to Proposition 2.2 in \cite{Vin}, there exists a canonical homology basis such that
\begin{equation}
\left(\begin{array}{cc}
\tau \mathbf{\mathcal{A}}\\
\tau \mathbf{\mathcal{B}}
\end{array}\right)
=
\left(\begin{array}{cr}
\mathbb{I}_{g}&0\\
\mathbb{H}&-\mathbb{I}_{g}
\end{array}\right)
\left(\begin{array}{cc}
\mathbf{\mathcal{A}}\\
\mathbf{\mathcal{B}}
\end{array}\right), \label{hom basis}
\end{equation}
where $\mathbb{I}_{g}$ is the $g\times g$ unit matrix, and $\mathbb{H}$ is a $g\times g$ matrix defined as follows
\\\\
1) if $\mathcal{R}_{g}(\R)\neq \emptyset$, 
\[\mathbb{H}={\left(\begin{array}{cccccccc}
0&1&&&&&&\\
1&0&&&&&&\\
&&\ddots&&&&&\\
&&&0&1&&&\\
&&&1&0&&&\\
&&&&&0&&\\
&&&&&&\ddots&\\
&&&&&&&0
\end{array}\right)} \quad \text{if $\mathcal{R}_{g}(\R)$ is dividing},\]
\[\mathbb{H}={\left(\begin{array}{cccccccc}
1&&&&&\\
&\ddots&&&&\\
&&1&&&\\
&&&0&&\\
&&&&\ddots&\\
&&&&&0
\end{array}\right)} \quad \text{if $\mathcal{R}_{g}(\R)$ is non-dividing},\]
(rank$(\mathbb{H})=g+1-k$ in both cases).
\\\\
2) if $\mathcal{R}_{g}(\R)= \emptyset$, (i.e. the curve does not have real ovals), then
\[\mathbb{H}={\left(\begin{array}{cccccc}
0&1&&&\\
1&0&&&\\
&&\ddots&&\\
&&&0&1\\
&&&1&0
\end{array}\right)}
\quad \text{or} \quad 
\mathbb{H}={\left(\begin{array}{ccccccc}
0&1&&&&\\
1&0&&&&\\
&&\ddots&&&\\
&&&0&1&\\
&&&1&0&\\
&&&&&0
\end{array}\right)},\]
(rank$(\mathbb{H})=g$ if $g$ is even, rank$(\mathbb{H})=g-1$ if $g$ is odd).
\\

Let us choose the homology basis satisfying (\ref{hom basis}), and study the action of $\tau$ on the normalized holomorphic differentials, and the action of the complex conjugation on the theta function with zero characteristics. 

By (\ref{hom basis}) the $\mathcal{A}$-cycles of the homology basis are invariant under $\tau$.  Due to normalization condition (2.1) this leads to the following action of $\tau$ on the normalized holomorphic differentials 
\begin{equation}
\overline{\tau^{*}\omega_{j}}=-\omega_{j}. \label{diff hol}
\end{equation}
Using (\ref{hom basis}) and (\ref{diff hol}) we get the following reality property for the matrix $\mathbb{B}$ of $\mathcal{B}$-periods
\begin{equation}
\overline{\mathbb{B}}= \mathbb{B}-2\mathrm{i}\pi\, \mathbb{H}. \label{matrix B}
\end{equation}
By Proposition 2.3 in \cite{Vin}, for any $\mathbf{z}\in\C^{g}$, relation (\ref{matrix B}) implies
\begin{equation} 
\overline{\Theta(\mathbf{z})}=\kappa\,\Theta(\overline{\mathbf{z}}-\mathrm{i}\pi \,\text{diag}(\mathbb{H})), \label{conj theta}
\end{equation}
where $\text{diag}(\mathbb{H})$ denotes the vector of the diagonal elements of the matrix $\mathbb{H}$, and $\kappa$ is a root of unity which depends on matrix $\mathbb{H}$ (knowledge of the exact value of $\kappa$ is not needed for our purpose).

\subsection{Action of $\tau$ on $H_{1}(\Rs_{g}\setminus\{a,b\})$ and $H_{1}(\Rs_{g},\{a,b\})$}

Here, we study the action of $\tau$ on the homology group  $H_{1}(\Rs_{g}\setminus\{a,b\})$ of the punctured Riemann surface $\Rs_{g}\setminus\{a,b\}$, and the action of $\tau$ on its dual relative homology group $H_{1}(\Rs_{g},\{a,b\})$. We consider the case where $\tau a=b$, and the case where $\tau a=a$, $\tau b=b$.

Denote by $\left\{\mathbf{\mathcal{A}},\mathbf{\mathcal{B}},\ell\right\}$ the generators of the relative homology group $H_{1}(\Rs_{g},\{a,b\})$, where $\ell$ is a contour between $a$ and $b$ which does not intersect the canonical homology basis $\left\{\mathbf{\mathcal{A}},\mathbf{\mathcal{B}}\right\}$, and denote by $\left\{\mathbf{\mathcal{A}},\mathbf{\mathcal{B}},\mathcal{S}_{b}\right\}$ the generators of the homology group $H_{1}(\Rs_{g}\setminus\{a,b\})$, where $\mathcal{S}_{b}$ is a positively oriented small contour around $b$ such that $\mathcal{S}_{b}\circ\ell=1$.

\subsubsection{Case $\tau a=b$}

\begin{proposition} Let us choose the canonical homology basis in $H_{1}(\Rs_{g})$ satisfying (\ref{hom basis}), and assume that $\tau a=b$. Then
\begin{enumerate}
	\item the action of $\tau$ on the generators $\left\{\mathbf{\mathcal{A}},\mathbf{\mathcal{B}},\ell\right\}$ of the relative homology group $H_{1}(\Rs_{g},\{a,b\})$ is given by
\begin{equation}
\left(\begin{array}{ccc}
\tau \mathbf{\mathcal{A}}\\
\tau \mathbf{\mathcal{B}}\\
\tau \ell
\end{array}\right)
=
\left(\begin{array}{ccc}
\mathbb{I}_{g}&0&0\\
\mathbb{H}&-\mathbb{I}_{g}&0\\
\mathbf{N}&0&-1
\end{array}\right)
\left(\begin{array}{ccc}
\mathbf{\mathcal{A}}\\
\mathbf{\mathcal{B}} \\
\ell
\end{array}\right), \label{hom basis 1}
\end{equation}
for some $\mathbf{N}\in\Z^{g}$,
\item the action of $\tau$ on the generators $\left\{\mathbf{\mathcal{A}},\mathbf{\mathcal{B}},\mathcal{S}_{b}\right\}$ of the homology group $H_{1}(\Rs_{g}\setminus\{a,b\})$ is given by
\begin{equation}
\left(\begin{array}{ccc}
\tau \mathbf{\mathcal{A}}\\
\tau \mathbf{\mathcal{B}}\\
\tau \mathcal{S}_{b}
\end{array}\right)
=
\left(\begin{array}{ccc}
\mathbb{I}_{g}&0&0\\
\mathbb{H}&-\mathbb{I}_{g}&\mathbf{N}\\
0&0&1
\end{array}\right)
\left(\begin{array}{ccc}
\mathbf{\mathcal{A}}\\
\mathbf{\mathcal{B}} \\
\mathcal{S}_{b}
\end{array}\right), \label{hom basis 2}
\end{equation}
\end{enumerate}
where vector $\mathbf{N}\in\Z^{g}$ is the same as in (\ref{hom basis 1}).
\end{proposition}

\begin{proof}
The action of $\tau$ on $\mathcal{A}$ and $\mathcal{B}$-cycles in (\ref{hom basis 1}) coincides with the one  (\ref{hom basis}) in $H_{1}(\Rs_{g})$.
From (\ref{hom basis}), one sees that any contour in $H_{1}(\Rs_{g})$ which is invariant under $\tau$ is a combination of $\mathcal{A}$-cycles only. In particular, the closed contour $\tau \ell+\ell\in H_{1}(\Rs_{g})$ can be written as
\begin{equation}
\tau \ell+\ell= \mathbf{N}\mathcal{A},  \label{C}
\end{equation}
for some $\mathbf{N}\in\Z^{g}$. This proves (\ref{hom basis 1}).
\\

Now let us prove (\ref{hom basis 2}). By (\ref{hom basis}), the cycles $\tau \mathcal{A}$ admit the following decomposition in $H_{1}(\Rs_{g}\setminus\{a,b\})$:
\begin{equation}
\tau \mathcal{A}= \mathcal{A}+\mathbf{n}\,\mathcal{S}_{b},  \label{A S}
\end{equation}
for some $\mathbf{n}\in\Z^{g}$.
Since $\tau$ changes the orientation of $\Rs_{g}$, all intersection indices change their sign under the action of $\tau$. We get from (\ref{A S}) 
\begin{align}
0&=\mathcal{A}\circ\ell \nonumber\\
&=-\tau \mathcal{A}\circ\tau \ell \nonumber\\
&=-(\mathcal{A}+\mathbf{n}\,\mathcal{S}_{b})\circ \tau \ell \nonumber\\
&=-(\mathcal{A}+\mathbf{n}\,\mathcal{S}_{b})\circ (-\ell+\mathbf{N}\mathcal{A}), \label{int ind A1}
\end{align}
where $\mathbf{N}\in\Z^{g}$ is defined by (\ref{hom basis 1}). The last intersection index in (\ref{int ind A1}) equals $\mathbf{n}$, which implies $\tau \mathcal{A}=\mathcal{A}$. According to (\ref{hom basis}), the action of $\tau$ on $\mathcal{B}$-cycles in $H_{1}(\Rs_{g}\setminus\{a,b\})$ is given by
\begin{equation}
\tau \mathcal{B}= -\mathcal{B}+\mathbb{H}\mathcal{A}+\mathbf{m}\,\mathcal{S}_{b},  \label{B S}
\end{equation}
for some $\mathbf{m}\in\Z^{g}$. Then
\begin{align}
0&=\mathcal{B}\circ\ell \nonumber\\
&=-\tau \mathcal{B}\circ\tau \ell \nonumber\\
&=-(-\mathcal{B}+\mathbb{H}\mathcal{A}+\mathbf{m}\,\mathcal{S}_{b})\circ \tau \ell \nonumber\\
&=-(-\mathcal{B}+\mathbb{H}\mathcal{A}+\mathbf{m}\,\mathcal{S}_{b})\circ (-\ell+\mathbf{N}\mathcal{A}), \label{int ind B1}
\end{align}
where $\mathbf{N}$ is defined by (\ref{hom basis 1}).
The last intersection index in (\ref{int ind B1}) equals $\mathbf{m}-\mathbf{N}$, which gives $\tau \mathcal{B}= -\mathcal{B}+\mathbb{H}\mathcal{A}+\mathbf{N}\,\mathcal{S}_{b}$. Finally, to prove that $\tau \mathcal{S}_{b}=\mathcal{S}_{b}$, we use the relation $\mathcal{S}_{a}+\mathcal{S}_{b}=0$, where $\mathcal{S}_{a}$ is a positively oriented small contour around $a$, and the relation $\tau \mathcal{S}_{b}=-\mathcal{S}_{a}$. 
\end{proof}

\subsubsection{Case $\tau a=a$ and $\tau b=b$}

\begin{proposition} Let us choose the canonical homology basis in $H_{1}(\Rs_{g})$ satisfying (\ref{hom basis}), and assume that $\tau a=a$ and $\tau b=b$. Then
\begin{enumerate}
	\item the action of $\tau$ on the generators $\left\{\mathbf{\mathcal{A}},\mathbf{\mathcal{B}},\ell\right\}$ of the relative homology group $H_{1}(\Rs_{g},\{a,b\})$ is given by
\begin{equation}
\left(\begin{array}{ccc}
\tau \mathbf{\mathcal{A}}\\
\tau \mathbf{\mathcal{B}}\\
\tau \ell
\end{array}\right)
=
\left(\begin{array}{ccc}
\mathbb{I}_{g}&0&0\\
\mathbb{H}&-\mathbb{I}_{g}&0\\
\mathbf{N}&\mathbf{M}&1
\end{array}\right)
\left(\begin{array}{ccc}
\mathbf{\mathcal{A}}\\
\mathbf{\mathcal{B}} \\
\ell
\end{array}\right), \label{hom basis 3}
\end{equation}
where $\mathbf{N},\,\mathbf{M}\in\Z^{g}$ are related by
\begin{equation}
2\,\mathbf{N}+\mathbb{H}\mathbf{M}=0, \label{NM stable}
\end{equation}
\item the action of $\tau$ on the generators $\left\{\mathbf{\mathcal{A}},\mathbf{\mathcal{B}},\mathcal{S}_{b}\right\}$ of the homology group $H_{1}(\Rs_{g}\setminus\{a,b\})$ is given by
\begin{equation}
\left(\begin{array}{ccc}
\tau \mathbf{\mathcal{A}}\\
\tau \mathbf{\mathcal{B}}\\
\tau \mathcal{S}_{b}
\end{array}\right)
=
\left(\begin{array}{ccc}
\mathbb{I}_{g}&0&-\mathbf{M}\\
\mathbb{H}&-\mathbb{I}_{g}&\mathbf{N}\\
0&0&-1
\end{array}\right)
\left(\begin{array}{ccc}
\mathbf{\mathcal{A}}\\
\mathbf{\mathcal{B}} \\
\mathcal{S}_{b}
\end{array}\right), \label{hom basis 4}
\end{equation}
\end{enumerate}
where vectors $\mathbf{N},\,\mathbf{M}\in\Z^{g}$ are the same as in (\ref{hom basis 3}).
\end{proposition}

\begin{proof}
The action of $\tau$ on $\mathcal{A}$ and $\mathcal{B}$-cycles in (\ref{hom basis 3}) coincides with the one (\ref{hom basis}) in $H_{1}(\Rs_{g})$.
From (\ref{hom basis}), one sees that each contour $\mathcal{C}\in H_{1}(\Rs_{g})$ which satisfies $\tau \mathcal{C}=-\mathcal{C}$, can be represented by
\begin{equation}
\mathcal{C}= \mathbf{\tilde{N}}\mathcal{A}+\mathbf{\tilde{M}}\mathcal{B},  \label{C stable}
\end{equation}
where $\mathbf{\tilde{N}},\,\mathbf{\tilde{M}}\in\Z^{g}$ are related by $2\,\mathbf{\tilde{N}}+\mathbb{H}\mathbf{\tilde{M}}=0$. In particular, the closed contour $\tau \ell-\ell\in H_{1}(\Rs_{g},\{a,b\})$ can be written as
\begin{equation}
\mathcal{\tau \ell-\ell}= \mathbf{N}\mathcal{A}+\mathbf{M}\mathcal{B},  \label{tau ell stable}
\end{equation}
where $\mathbf{N},\,\mathbf{M}\in\Z^{g}$ are related by $2\,\mathbf{N}+\mathbb{H}\mathbf{M}=0$.
This proves (\ref{hom basis 3}).
\\

Now let us prove (\ref{hom basis 4}). By (\ref{hom basis}), the cycles $\tau \mathcal{A}$ admit the following decomposition in $H_{1}(\Rs_{g}\setminus\{a,b\})$
\begin{equation}
\tau \mathcal{A}= \mathcal{A}+\mathbf{n}\,\mathcal{S}_{b},  \label{A S stable}
\end{equation}
for some $\mathbf{n}\in\Z^{g}$.
Therefore, we get from (\ref{A S stable}) 
\begin{align}
0&=\mathcal{A}\circ\ell \nonumber\\
&=-\tau \mathcal{A}\circ\tau \ell \nonumber\\
&=-(\mathcal{A}+\mathbf{n}\,\mathcal{S}_{b})\circ \tau \ell \nonumber\\
&=-(\mathcal{A}+\mathbf{n}\,\mathcal{S}_{b})\circ (\ell+\mathbf{N}\mathcal{A}+\mathbf{M}\mathcal{B}), \label{int ind A2}
\end{align}
where $\mathbf{N},\,\mathbf{M}\in\Z^{g}$ is defined by (\ref{hom basis 3}).
The last intersection index in (\ref{int ind A2}) equals $-(\mathbf{n}+\mathbf{M})$, which gives $\tau \mathcal{A}= \mathcal{A}-\mathbf{M}\,\mathcal{S}_{b}$.
According to (\ref{hom basis}), the action of $\tau$ on $\mathcal{B}$-cycles in $H_{1}(\Rs_{g}\setminus\{a,b\})$ is given by
\begin{equation}
\tau \mathcal{B}= -\mathcal{B}+\mathbb{H}\mathcal{A}+\mathbf{m}\,\mathcal{S}_{b}, \label{B S stable}
\end{equation}
for some $\mathbf{m}\in\Z^{g}$.
Then
\begin{align}
0&=\mathcal{B}\circ\ell \nonumber\\
&=-\tau \mathcal{B}\circ\tau \ell \nonumber\\
&=-(-\mathcal{B}+\mathbb{H}\mathcal{A}+\mathbf{m}\,\mathcal{S}_{b})\circ \tau \ell \nonumber\\
&=-(-\mathcal{B}+\mathbb{H}\mathcal{A}+\mathbf{m}\,\mathcal{S}_{b})\circ (\ell+\mathbf{N}\mathcal{A}+\mathbf{M}\mathcal{B}), \label{int ind B2}
\end{align}
where $\mathbf{N},\,\mathbf{M}\in\Z^{g}$ are defined by (\ref{hom basis 3}).
The last intersection index in (\ref{int ind B2}) equals $-(\mathbf{m}+\mathbf{N}+\mathbb{H}\mathbf{M})$, which by (\ref{NM stable}) implies $\tau \mathcal{B}= -\mathcal{B}+\mathbb{H}\mathcal{A}+\mathbf{N}\,\mathcal{S}_{b}$.
Finally, since the anti-holomorphic involution $\tau$ inverses orientation we have $\tau \mathcal{S}_{b}=-\mathcal{S}_{b}$. This completes the proof of Proposition A.1.
\end{proof}

\subsection{Action of $\tau$ on the Jacobian and theta divisor of real Riemann surfaces}

In this part, we review known results \cite{Vin}, \cite{DN} about the theta divisor of real Riemann surfaces. Let us choose the canonical homology basis satisfying (\ref{hom basis}) and consider the Jacobian $J=J(\Rs_{g})$ of the real Riemann surface $\Rs_{g}$.
The Abel map (\ref{abel}) $\mu:\Rs_{g}\longmapsto J$ can be extended linearly to all divisors on $\Rs_{g}$, which defines a map on linear equivalence classes of divisors.

The anti-holomorphic involution $\tau$ on $\Rs_{g}$ gives rise 
to an anti-holomorphic involution on the Jacobian
$J$: 
if $\mathcal{D}$ is a positive divisor of degree $n$ on $\Rs_{g}$, then $\tau\,\mathcal{D}$ is the class of the point $(\int_{n\,\tau 
p_{0}}^{\tau 
\mathcal{D}}\omega)=(\int_{n\,p_{0}}^{\mathcal{D}}\tau^{*}\omega)$ in the Jacobian. Therefore  by (\ref{diff hol}), $\tau$ lifts to the anti-holomorphic involution on $J$, denoted also by $\tau$, given by 
\begin{equation}
\tau \zeta=-\overline{\zeta}+n_{\zeta}\,\mu(\tau p_{0}), \quad \forall \zeta\in J,   \label{3.7}
\end{equation}
where $n_{\zeta}\in\Z$, $n_{\zeta}\leq g$, is the degree of the divisor $\mathcal{D}$ such that $\mu(\mathcal{D})=\zeta$.

Now consider the following two subsets of the Jacobian
\begin{align}
S_{1}=\{\zeta\in J;\, \zeta+\tau\,\zeta=\mathrm{i}\pi \,\text{diag}(\mathbb{H})\},  \label{S1}\\
S_{2}=\{\zeta\in J;\, \zeta-\tau\,\zeta=\mathrm{i}\pi \,\text{diag}(\mathbb{H})\}. \label{S2}
\end{align}
In this section we study their intersections $S_{1}\cap (\Theta)$ and 
$S_{2}\cap (\Theta)$ with the theta divisor $(\Theta)$, the set of zeros of the theta function. 

Let us introduce the following notations: $(e_{i})_{k}=\delta_{ik}$, $\mathbb{B}_{i}=\mathbb{B}\,e_{i}$. The following proposition was proved in \cite{Vin}.

\begin{proposition} The set $S_{1}$ is a disjoint union of the tori 
$T_{v}$ defined by
\begin{align}
T_{v}=\{\zeta\in J;\, \zeta=2\mathrm{i}\pi\,(\tfrac{1}{4}\,\text{diag}(\mathbb{H})+\tfrac{v_{1}}{2}\,e_{r+1}+\ldots+\tfrac{v_{g-r}}{2}\,e_{g}) 
+\beta_{1}\,\text{Re}(\B_{1})+\ldots+\beta_{g}\,\text{Re}(\B_{g})\,, \nonumber
\\
\beta_{1},\ldots,\beta_{r}\in\R/2\Z\,,\,\beta_{r+1},\ldots,\beta_{g}\in\R/\Z\}, \label{3.10}
\end{align}
where $v=(v_{1},\ldots,v_{g-r})\in(\Z/2\Z)^{g-r}$ and $r$ is the rank of the matrix $\mathbb{H}$. Moreover, if $\mathcal{R}_{g}(\R)\neq\emptyset$, then $T_{v}\cap (\Theta)=\emptyset$ if and only if the curve is dividing and $v=0$.
\end{proposition}

The last statement means that among all curves which admit real ovals, the only torus $T_v$ which does not intersect the theta-divisor is the  torus $T_{0}$ corresponding to dividing curves. 
This torus is given by 
\begin{equation}
T_{0}=\{\zeta\in J; \,\zeta=\beta_{1}\,\text{Re}(\B_{1})+\ldots+\beta_{g}\,\text{Re}(\B_{g}),
\,\beta_{1},\ldots,\beta_{r}\in\R/2\Z\,,\,\beta_{r+1},\ldots,\beta_{g}\in\R/\Z \}.  \label{3.11}
\end{equation}
\\
The following proposition was proved in \cite{DN}.

\begin{proposition} The set $S_{2}$ is a disjoint union of the tori 
$\tilde{T}_{v}$ defined by
\begin{equation}
\tilde{T}_{v}=\{\zeta\in J\,;\, \zeta=2\mathrm{i}\pi\left(\alpha_{1}\,e_{1}+\ldots+\alpha_{g}\,e_{g}\right)+\tfrac{v_{1}}{2}\,\B_{r+1}+\ldots+\tfrac{v_{g-r}}{2}\,\B_{g},
\alpha_{1},\ldots,\alpha_{g}\in\R/\Z\}, \label{3.12}
\end{equation}
where $v=(v_{1},\ldots,v_{g-r})\in(\Z/2\Z)^{g-r}$ and $r$ is the rank of the matrix $\mathbb{H}$. Moreover, if $\mathcal{R}_{g}(\R)\neq\emptyset$, then $\tilde{T}_{v}\cap (\Theta)=\emptyset$ if and only if the curve is an M-curve and $v=0$.
\end{proposition}

\section{Computation of the argument of the fundamental scalar $q_{2}(a,b)$}

This section is devoted to the computation of $\arg\{q_{2}(a,b)\}$, where $q_{2}(a,b)$ is defined by (\ref{q2}).
As before, $\Rs_{g}$ denotes a real compact Riemann surface of genus $g$ with an anti-holomorphic involution $\tau$. The argument of $q_{2}(a,b)$ is computed both in the case $\tau a=b$, as well as in the case $\tau a=a$, $\tau b=b$.

\subsection{Integral representation for $q_{2}(a,b)$}

Assume that $a,b\in\Rs_{g}$ can be connected by a contour which does not intersect basic cycles. Hence we can define the normalized meromorphic differential of the third kind $\Omega_{b-a}$ which has residue $1$ at $b$ 
and residue $-1$ at $a$.

\vspace{0.3cm}
\begin{proposition} Let $a,b$ be distinct points on a compact Riemann surface $\Rs_{g}$ of genus $g$. Denote by $k_{a}$ and $k_{b}$ local parameters in a neighbourhood of $a$ and $b$ respectively. Then the quantity $q_{2}(a,b)$ defined in 
(\ref{q2}) admits the following integral representation
\begin{equation}
q_{2}(a,b)=-\lim_{\begin{array}{cc}
\tilde{b}\rightarrow b \\
\tilde{a}\rightarrow a
\end{array}}
\left[\left(k_{a}(\tilde{a})\,k_{b}(\tilde{b})\right)^{-1}\exp\left\{\int_{\tilde{a}}^{\tilde{b}}\Omega_{b-a}(p)\right\}\right], \label{q2 int}
\end{equation}
where the integration contour between $\tilde{a}$ and $\tilde{b}$, which in the sequel is denoted by $\tilde{\ell}$, does not cross any cycle from the canonical homology basis.
\end{proposition}

\begin{proof} Notice that the scalar $q_{2}(a,b)$ does not depend on the choice of the contour $\tilde{\ell}$, assuming that $\tilde{\ell}$ lies in the fundamental polygon of the Riemann surface.

Denote by $k_{x}$ a local parameter in a neighbourhood of a point $x\in\Rs_{g}$. To prove (\ref{q2 int}), recall that
\begin{equation}
\int_{\tilde{a}}^{\tilde{b}}\Omega_{b-a}(p)=\ln\frac{\Theta[\delta](\int_{b}^{\tilde{b}})}{\Theta[\delta](\int_{a}^{\tilde{b}})}+\ln\frac{\Theta[\delta](\int_{a}^{\tilde{a}})}{\Theta[\delta](\int_{b}^{\tilde{a}})}.  \label{def diff3}
\end{equation}
Since $\delta$ is an odd non singular characteristic, the expression $\frac{\Theta[\delta](\int_{b}^{p})}{\Theta[\delta](\int_{a}^{p})}$ has a simple zero at $b$ and a simple pole at $a$. Therefore, if we consider $\tilde{a}$ lying in a neighbourhood of $a$, and $\tilde{b}$ lying in a neighbourhood of $b$, we get (with $\alpha_{1},\beta_{1}\neq 0$)
\begin{eqnarray}
\frac{\Theta[\delta](\int_{b}^{\tilde{b}})}{\Theta[\delta](\int_{a}^{\tilde{b}})}=\alpha_{1}\,k_{b}(\tilde{b})+o(k_{b}(\tilde{b})), \label{alpha1}\\
\frac{\Theta[\delta](\int_{a}^{\tilde{a}})}{\Theta[\delta](\int_{b}^{\tilde{a}})}=\beta_{1}\,k_{a}(\tilde{a})+o(k_{a}(\tilde{a})).\label{alpha2}
\end{eqnarray}
Combining (\ref{def diff3}) together with (\ref{alpha1}) and (\ref{alpha2}), we obtain the following relation
\begin{equation}
\lim_{\begin{array}{cc}
\tilde{b}\rightarrow b \\
\tilde{a}\rightarrow a
\end{array}}
\left[\left(k_{a}(\tilde{a})\,k_{b}(\tilde{b})\right)^{-1}\exp\left\{\int_{\tilde{a}}^{\tilde{b}}\Omega_{b-a}(p)\right\}\right]=\alpha_{1}\beta_{1}. \label{q2 lim}
\end{equation}
Moreover, using the definition (\ref{q2}) of $q_{2}(a,b)$, it follows from (\ref{alpha1}) and (\ref{alpha2}) that $\alpha_{1}\beta_{1}=-q_{2}(a,b)$,
which by (\ref{q2 lim}) completes the proof. 
\end{proof}

\subsection{Argument of $q_{2}(a,b)$ when $\tau a=b$} 
Here we compute the argument of the fundamental scalar $q_{2}(a,b)$ defined in (\ref{q2}) in the case where $\tau a=b$.
Let us choose the homology basis satisfying (\ref{hom basis}). 

\vspace{0.3cm}
\begin{proposition} Let $a,b\in\Rs_{g}$ be distinct points such that $\tau a=b$, with local parameters satisfying the relation $\overline{k_{b}(\tau p)}=k_{a}(p)$ for any point $p$ lying in a neighbourhood of $a$. Consider a contour $\ell$ connecting points $a$ and $b$; assume that $\ell$ is lying in the fundamental polygon of the Riemann surface $\Rs_{g}$. Then the scalar $q_{2}(a,b)$ is real, and its sign is given by:
\begin{enumerate}
\item if $\ell$ intersects the set of real ovals of $\Rs_{g}$ only once, and if this intersection is transversal, then $q_{2}(a,b)<0$,
\item if $\ell$ does not cross any real oval, then $q_{2}(a,b)>0$.
\end{enumerate}
\end{proposition}

\begin{proof}
Let $\tilde{a},\tilde{b}\in\Rs_{g}$ lie in a neighbourhood of $a$ and $b$ respectively, and $\tau \tilde{a}=\tilde{b}$. Denote by $\tilde{\ell}$ an oriented contour connecting $\tilde{a}$ and $\tilde{b}$. First, let us check that
\begin{equation}
\arg\{q_{2}(a,b)\}= \pi(1+\alpha),  \label{argum1 q2}
\end{equation}
where $\alpha=(\tau \tilde{\ell}+\tilde{\ell})\circ \ell$.
The integral representation (\ref{q2 int}) of $q_{2}(a,b)$ leads to
\begin{equation}
\arg\{q_{2}(a,b)\}=  \pi + \text{Im}\left(\int_{\tilde{\ell}}\Omega_{b-a}(p)\right). \label{argum2 q2}
\end{equation}
Using the action (\ref{hom basis 1}) of $\tau$ on the $\mathcal{A}$-cycles in the homology group $H_{1}(\Rs_{g}\setminus\{a,b\})$, we get the following action of $\tau$ on the normalized meromorphic differentials of third kind $\Omega_{b-a}$:
\begin{equation}
\overline{\tau^{*}\Omega_{b-a}}=-\Omega_{ b -  a}, \label{diff 3kind}
\end{equation}
(notice that $\tau a=b$). Hence, the last term in the right hand side of (\ref{argum2 q2}) is equal to $\frac{1}{2\mathrm{i}}\int_{\tau \tilde{\ell}+\tilde{\ell}}\Omega_{b-a}(p)$. 
The closed contour $\tau \tilde{\ell}+\tilde{\ell}$ admits the following decomposition in $H_{1}(\Rs_{g}\setminus\{a,b\})$,
\begin{equation}
\tau \tilde{\ell}+\tilde{\ell}=\mathbf{N}\mathcal{A}+\alpha\,\mathcal{S}_{b}, \label{alpha B}
\end{equation}
where $\alpha=(\tau \tilde{\ell}+\tilde{\ell})\circ \ell$ and $\mathbf{N}\in\Z^{g}$ is defined in (\ref{hom basis 3}). Since the differential $\Omega_{b-a}$ has vanishing $\mathcal{A}$-periods, by (\ref{alpha B}) we obtain
\begin{equation}
\int_{\tau \tilde{\ell}+\tilde{\ell}}\Omega_{b-a}(p)= 2\mathrm{i}\pi\alpha, \label{pol per1}
\end{equation}
which leads to (\ref{argum1 q2}). Therefore, the sign of $q_{2}(a,b)$ depends on the parity of the intersection index $\alpha=(\tau \tilde{\ell}+\tilde{\ell})\circ \ell$.
\\\\
Let us now consider cases (1) and (2) separatly.
\\\\
\textit{Case (1).} Assume that each of the contours $\ell$ and $\tilde{\ell}$ intersects the set of real ovals of $\Rs_{g}$ transversally only once, and, moreover, this intersection point is the same for $\ell$ and $\tilde{\ell}$; we denote it by $p_{0}\in\Rs_{g}(\R)$.
Then the closed contour $\tau \tilde{\ell}+\tilde{\ell}$ can be decomposed into a sum of two closed contours $c\tilde{\ell_{1}}$ and $c\tilde{\ell_{2}}$, having the common point $p_{0}$, and such that $\tau$ sends the set of points $\left\{c\tilde{\ell_{1}}\right\}$ into the set of points
$\left\{c\tilde{\ell_{2}}\right\}$. Therefore, if the orientation of $c\tilde{\ell_{1}}$ and $c\tilde{\ell_{2}}$ is inherited from the orientation of $\tau \tilde{\ell}+\tilde{\ell}$, we have $\tau c\tilde{\ell_{1}}= c\tilde{\ell_{2}}$ as elements of $H_{1}(\Rs_{g}\setminus\{a,b\})$. Then,
\begin{equation}
c\tilde{\ell_{1}}\circ\ell=-\tau c\tilde{\ell_{1}}\circ\tau \ell=-c\tilde{\ell_{2}}\circ(-\ell+\mathcal{A}\mathbf{N})
=c\tilde{\ell_{2}}\circ\ell, \nonumber
\end{equation}
where we used the action (\ref{hom basis 1}) of $\tau$ on the contour $\ell$, and the fact that the intersection index between $c\tilde{\ell_{2}}$ and $\mathcal{A}$-cycles is zero by (\ref{alpha B}). Hence the intersection index $\alpha$ satisfies
\[\alpha=(\tau \tilde{\ell}+\tilde{\ell})\circ\ell=(c\tilde{\ell_{1}}+c\tilde{\ell_{2}})\circ\ell=2,\]
which by (\ref{argum1 q2}) leads to $q_{2}(a,b)<0$.
\\\\
\textit{Case (2).} Let $\mathcal{V}$ be a ring neighbourhood of the path $\tau \tilde{\ell}+\tilde{\ell}$, bounded by two closed paths denoted by $\partial\mathcal{V}_{1}$ and $\partial\mathcal{V}_{2}$, in such way that the path $\ell$ lies in $\mathcal{V}$ and $\tau\left\{\partial\mathcal{V}_{1}\right\}=\left\{\partial\mathcal{V}_{2}\right\}$. We assume that $\mathcal{V}$ is chosen such that no point of $\mathcal{V}$ is invariant under $\tau$. Then $\mathcal{V}$ can be decomposed into two connected components denoted by $\mathcal{V}_{1}$ and $\mathcal{V}_{2}$ as follows: $\mathcal{V}_{1}$ is bounded by $\partial\mathcal{V}_{1}$ and $\tau \tilde{\ell}+\tilde{\ell}$, and $\mathcal{V}_{2}$ is bounded by $\partial\mathcal{V}_{2}$ and $\tau \tilde{\ell}+\tilde{\ell}$. Then $\tau\mathcal{V}_{1}=\mathcal{V}_{2}$ since the set of points $\left\{\tau \tilde{\ell}+\tilde{\ell}\right\}$ is invariant under  $\tau$. In particular if $a\in\mathcal{V}_{1}$, then $b\in\mathcal{V}_{2}$. Thus the intersection index
$\alpha=(\tau \tilde{\ell}+\tilde{\ell})\circ\ell$ is odd, which leads to $q_{2}(a,b)>0$.
\end{proof}

\subsection{Argument of $q_{2}(a,b)$ when $\tau a=a$ and $\tau b=b$} 
Now let us consider the case where $a$ and $b$ are invariant with respect to $\tau$. 
\\\\
\textbf{Proposition B.3.}\textit{ Let $a,b\in\Rs_{g}(\R)$ with local parameters satisfying $\overline{k_{a}(\tau p)}=k_{a}(p)$ for any point $p$ lying in a neighbourhood of $a$ and $\overline{k_{b}(\tau p)}=k_{b}(p)$ for any point $p$ lying in a neighbourhood of $b$. Denote by $\left\{\mathbf{\mathcal{A}},\mathbf{\mathcal{B}},\ell\right\}$ the generators of the relative homology group $H_{1}(\Rs_{g},\{a,b\})$ (see Section A.2). Let $\tilde{a},\tilde{b}\in\Rs_{g}(\R)$ lie in a neighbourhood of $a$ and $b$ respectively, and denote by $\tilde{\ell}$ an oriented contour connecting $\tilde{a}$ and $\tilde{b}$. Then the argument of the scalar $q_{2}(a,b)$ is given by 
\begin{equation}
\arg\{q_{2}(a,b)\}= \arg\{k_{a}(\tilde{a})\,k_{b}(\tilde{b})\}+\pi\,\left(1+\alpha+\frac{1}{2}\left\langle \mathbb{H}\mathbf{M},\mathbf{M}\right\rangle\right)-\frac{1}{2\mathrm{i}}\left(\left\langle \B\,\mathbf{M},\mathbf{M}\right\rangle+2\,\left\langle \mathbf{r} ,\mathbf{M}\right\rangle\right), \label{arg q2}
\end{equation}
where $\alpha$ equals the intersection index $(\tau\tilde{\ell}-\tau\tilde{\ell})\circ \ell$. Here $\mathbf{r}=\int_{\ell}\omega$, and $\mathbf{M}\in\Z^{g}$ is defined in (\ref{hom basis 3}).}

\begin{proof}
From the integral representation (\ref{q2 int}) of $q_{2}(a,b)$ we get
\begin{equation}
\arg\{q_{2}(a,b)\}= \pi +\arg\{k_{a}(\tilde{a})\,k_{b}(\tilde{b})\}+  \text{Im}\left(\int_{\tilde{\ell}}\Omega_{b-a}(p)\right). \label{arg1 q2}
\end{equation}
Considering the action (\ref{hom basis 4}) of $\tau$ on the $\mathcal{A}$-cycles, due to the uniqueness of the normalized differential of the  third kind $\Omega_{b-a}$, we obtain
\begin{equation}
\overline{\tau^{*}\Omega_{b-a}}=\Omega_{b-a}+\sum_{k}M_{k} \,\omega_{k}, \label{diff 3kind real}
\end{equation}
where $\omega_{k}$ are the normalized holomorphic differentials. Therefore
\[\text{Im}\left(\int_{\tilde{\ell}}\Omega_{b-a}(p)\right)\equiv\frac{1}{2\mathrm{i}}\left(\int_{\tilde{\ell}}\Omega_{b-a}-\int_{\tau\tilde{\ell}}\Omega_{b-a}-\sum_{k}M_{k}\int_{\tau \tilde{\ell}}\omega_{k}\right).\]
The closed contour $\tau\tilde{\ell}-\tilde{\ell}\in H_{1}(\Rs_{g})$ satisfies $\tau (\tau\tilde{\ell}-\tilde{\ell})=-(\tau\tilde{\ell}-\tilde{\ell})$; thus by (\ref{C stable}) it has the following decomposition in $H_{1}(\Rs_{g}\setminus\{a,b\})$
\begin{equation}
\tau\tilde{\ell}-\tilde{\ell}= \mathbf{N}\mathcal{A}+\mathbf{M}\mathcal{B}+\alpha\,\mathcal{S}_{b},  \label{tau tilde stable}
\end{equation}
for some $\alpha\in\Z$, where $\mathbf{N},\,\mathbf{M}\in\Z^{g}$ are defined in (\ref{hom basis 3}).
Hence we get
\begin{equation}
\text{Im}\left(\int_{\tilde{\ell}}\Omega_{b-a}(p)\right)\equiv\frac{1}{2\mathrm{i}}\left(-\int_{\mathcal{B}\mathbf{M}}\Omega_{b-a}+2\mathrm{i}\pi\alpha-\sum_{k}M_{k}\left\{\int_{\tilde{\ell}}\omega_{k}+\sum_{j}(\mathbb{B}_{jk}-\mathrm{i}\pi \,\mathbb{H}_{jk})M_{j}\right\}\right), \label{prop 3.5 Im1}
\end{equation}
where we used the fact that the normalized differential $\Omega_{b-a}$ has vanishing $\mathcal{A}$-periods, and that the integral over the small contour $\mathcal{S}_{b}$ of the holomorphic differentials is zero. Since by definition the contour $\ell$ does not cross any cycles of the absolute homology basis, 
\begin{equation}
\int_{\mathcal{B}\mathbf{M}}\Omega_{b-a}=\left\langle \mathbf{M},\mathbf{r}\right\rangle. \label{prop 3.5 Im2}
\end{equation}
Hence we get
\begin{equation}
\text{Im}\left(\int_{\tilde{\ell}}\Omega_{b-a}(p)\right)\equiv\pi \alpha+\frac{\pi}{2}\left\langle \mathbb{H}\mathbf{M},\mathbf{M}\right\rangle-\frac{1}{2\mathrm{i}}\left(\left\langle \mathbf{M},\tilde{\mathbf{r}}+\mathbf{r}\right\rangle+\left\langle \B\mathbf{M},\mathbf{M}\right\rangle\right), \label{prop 3.5 Im3}
\end{equation}
where $\tilde{\mathbf{r}}=\int_{\tilde{\ell}}\omega$. 
Considering the limit when $\tilde{a}$ tends to $a$ and $\tilde{b}$ tends to $b$, we obtain (\ref{arg q2}).
\end{proof}

\end{document}